\documentclass[sigconf]{acmart}

\usepackage{textcomp}
\usepackage{xcolor}
\usepackage{url}

\usepackage{subfig}

\usepackage{graphicx}
\usepackage{amsthm}
\usepackage{multicol}
\usepackage{here}
\usepackage[linesnumbered,ruled,vlined]{algorithm2e}

\SetCommentSty{commentsty}
\DontPrintSemicolon
\usepackage{afterpage}
\usepackage{enumerate}

\def\WR{\stackrel{wr}{\to}}
\def\VORW{\stackrel{\ll(rw)}{\to}}
\def\VOWW{\stackrel{\ll(ww)}{\to}}

\newtheorem{definition}{Definition}

\newtheorem{theorem}{Theorem}

\begin{document}

\title{Scheduling Space Expander: An Extension of Concurrency Control for Data Ingestion Queries}


\author{Sho Nakazono}
\affiliation{%
    \institution{
        NTT Computer and Data Science Laboratories
    }
    \city{}
    \state{Tokyo}
    \country{Japan}
}
\email{syou.nakazono.nu@hco.ntt.co.jp}

\author{Hiroyuki Uchiyama}
\affiliation{%
    \institution{
        Recruit Co., Ltd.
    }
    \city{}
    \state{Tokyo}
    \country{Japan}
}
\email{hiroyuki-uchiyama@recruit.co.jp}

\author{Yasuhiro Fujiwara}
\affiliation{%
    \institution{
        NTT Communication Science Laboratories
    }
    \city{}
    \state{Kanagawa}
    \country{Japan}
}
\email{yasuhiro.hujiwara.kh@hco.ntt.co.jp}

\author{Hideyuki Kawashima}
\affiliation{%
    \institution{
        Faculty of Environment and Information Studies, Keio University
    }
    \city{}
    \state{Kanagawa}
    \country{Japan}
}
\email{river@sfc.keio.ac.jp}

\renewcommand{\shortauthors}{Sho, et al.}

\begin{abstract}
  With the continuing advances of sensing devices and IoT applications, database systems needs to process data ingestion queries that update the sensor data frequently.
To process data ingestion queries with transactional correctness, we propose a novel protocol extension method, scheduling space expander (SSE).
The key idea of SSE is that we can safely omit an update if the update becomes outdated and unnecessary.
SSE adds another control flow to conventional protocols to test the transactional correctness of an erasing version order, which assumes that a transactions' updates are all outdated and unnecessary.
In addition, we present an optimization of SSE called epoch-based SSE (ESSE), which generates, tests, and maintains the erasing version order more efficiently than SSE.
Our approach makes the performance of data ingestion queries more efficient.
Experimental results demonstrate that our ESSE extensions of Silo and MVTO improve 2.7$\times$ and 2.5$\times$ performance on the TATP benchmark on a 144-core machine, and the extensions achieved performance comparable to that of the original protocol for the TPC-C benchmark.

\end{abstract}

\pagestyle{plain}

\maketitle

\section{Introduction}
\label{sec:introduction}

\begin{figure}[t]
    \centerline{
        \subfloat{
            \includegraphics[width=0.5\textwidth,clip, trim={3cm 0 3cm 0}]{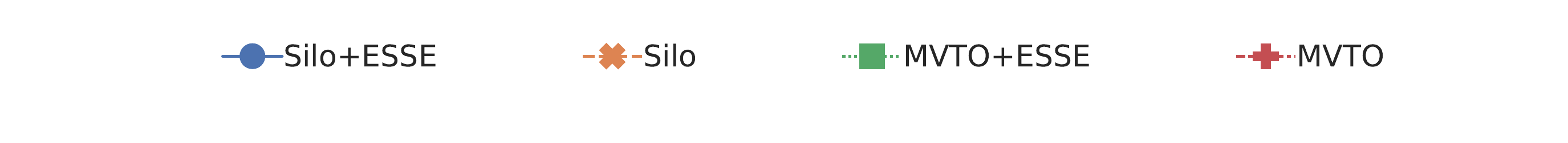}
        }
    }
    \vspace{-22pt}
    \addtocounter{subfigure}{-1}
    \centerline{
        \subfloat{
            \includegraphics[width=0.35\textwidth]{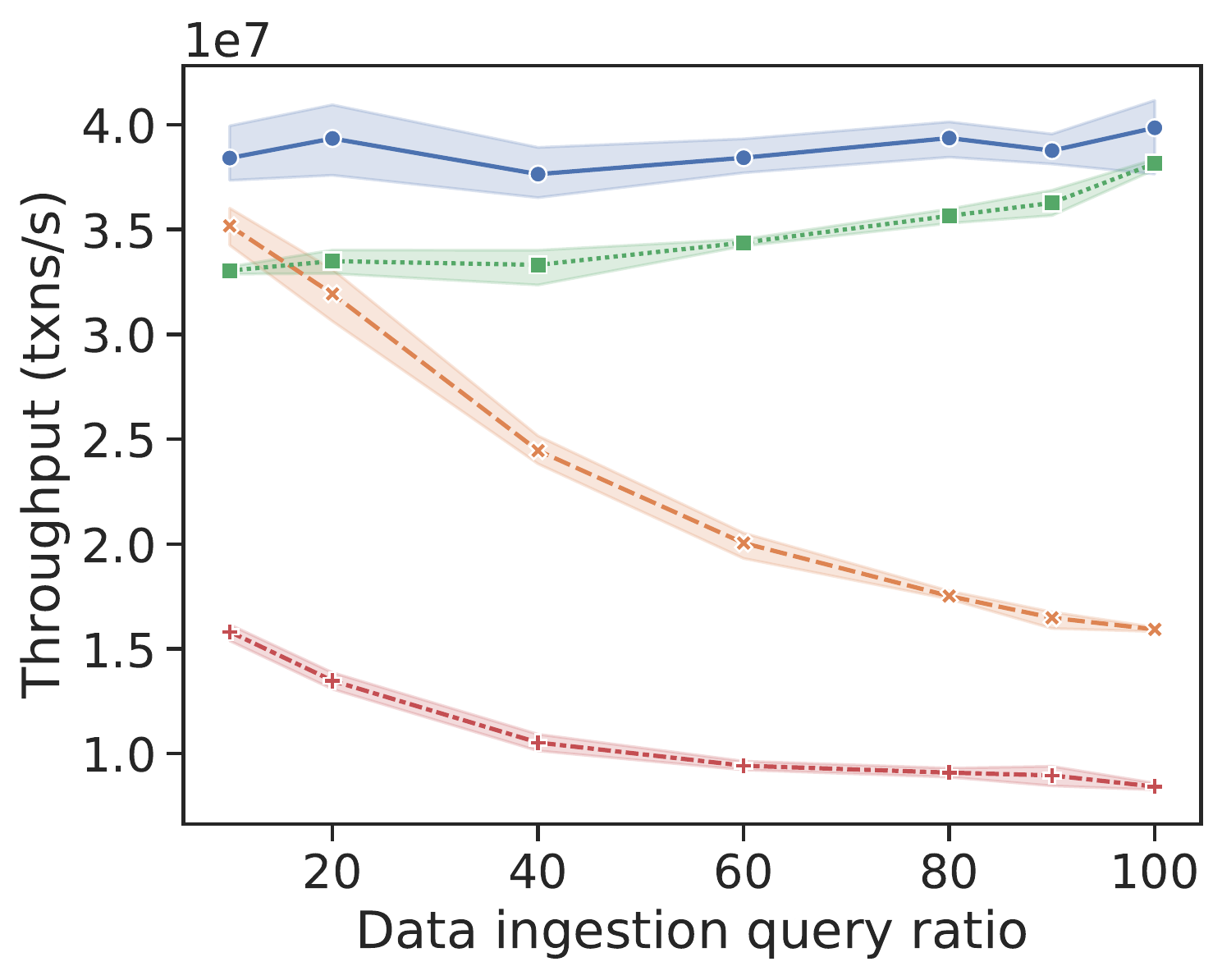}
        }
    }
    \vspace{-4pt}
    \caption{
        TATP benchmark throughput with respect to the data ingestion query ratio. The percentage of \texttt{UPDATE\_LOCATION} in TATP was varied to represent the data ingestion query ratio, and the throughput was measured with 144 worker threads. The original percentage is 14\%. The band of each line indicates the variance in five runs of the experiment.}
    \vspace{-8pt}

    \label{fig:TATP-VaryingDIQueryRatio}
\end{figure}

Modern internet of things (IoT) or mobile telecommunication services use billions of sensors, connected cars, or mobile devices to ingest the data from these sources continuously through a network.
These services read the real-time ingested data to operate real-world actuators such as the assembly factory machines, robotic highway construction markers, or location-based digital signage~\cite{DBLP:conf/ipsn/DemirbasSH08}.
This pair of data ingestion queries and real-time operations is becoming common process in network service providers.
These queries have long been operated in non-transactional systems such as streaming databases~\cite{DBLP:conf/edbt/GroverC15,DBLP:conf/vldb/TuLPY06}.
However, in recent years, the importance of \textit{the transactional correctness} has been studied~\cite{meehan2017data,DBLP:journals/pvldb/WangC19,10.1145/2882903.2899406}, since reading inconsistent or outdated data causes a malfunction, accident, or opportunity loss~\cite{DBLP:conf/ipsn/DemirbasSH08,meehan2017data, DBLP:conf/icc/TangGGLW11} in real-world actuators.

\begin{table*}[ht]
    \centering
    \begin{tabular}{cccc}
        \toprule
        Protocol                                              & Transactional correctness      & Omit write operations & Version storage \\
        \midrule
        Timestamp Ordering (T/O) with Thomas Write Rule (TWR) & Not strictly serializable      & \textbf{Yes}          & 1VCC            \\
        Silo OCC                                              & \textbf{Strictly serializable} & No                    & 1VCC            \\
        Cicada MVTO                                           & Not strictly serializable      & No                    & MVCC            \\
        Silo + ESSE                                           & \textbf{Strictly serializable} & \textbf{Yes}          & 1VCC            \\
        MVTO + ESSE                                           & \textbf{Strictly serializable} & \textbf{Yes}          & MVCC            \\
        \bottomrule
    \end{tabular}
    \caption{Differences between conventional protocols and our proposal}
    \label{tab:comparison}
\end{table*}

The Telecom Application Transaction Processing Benchmark (TATP)~\cite{TATP,DBLP:conf/sigmod/0001MK15} is an example of the transactional workload which includes data ingestion and real-time operation queries.
In TATP, tremendous mobile devices update the subscribers' current locations or profiles.
Concurrently, telecom base stations read the ingested data as the online transaction to operate the migration of Home Location Register databases.
A key feature of data ingestion queries in TATP is ~\textbf{blind updates}.
It is not read-modify-write; each mobile device updates its own location information, but it does not need to read the previous information written by the device itself.
In addition, it is not inserts; we cannot even handle the query as ``unique timestamped inserts''.
Although this approach is suitable for time-series analytical systems, for real-time operations, we need to guarantee that each data item has only one latest version, and that the consistent results hold for multiple data items~\cite{10.1145/2882903.2899406,meehan2017data,DBLP:conf/ipsn/DemirbasSH08}.

Blind update is the most significant difference between our intended applications and traditional applications.
Traditional workloads such as TPC-C~\cite{10.5555/1946050.1946051}, which models wholesale warehouse management, also generate a large number of writes in the form of new orders, but these workloads request inserts, not blind updates.
Inserts are scalable by partitioning since they write distinct items, but updates must write to the same data item in principle, and thus it is difficult to scale.
In our experiments, as the rate of data ingestion queries increases, existing protocols have exhibited degraded performance for TATP benchmark as shown in Figure~\ref{fig:TATP-VaryingDIQueryRatio}.
This is because they have to use a lock mechanism to order the update requests serially into the same data item, to preserve the transactional correctness.
Such serial execution of updates causes performance degradation.
If the throughput becomes less than the data velocity, we cannot operate the system and services of IoT/Telecom applications.

To process massive amounts of data in real-time, there exist methods to \textbf{omit} updates without lock mechanisms by using load shedding or backpressure\cite{DBLP:conf/vldb/TuLPY06,DBLP:journals/tmc/TahirYM18}.
However, these methods do not have enough guarantee of the transactional correctness.
For example, when a sensor updates information of two nearby moving objects but the system omits one of them partially, actuators only obtain one moving object data and thus some actuators might cause an accident such as a collision.
To guarantee the transactional correctness, databases need to use concurrency control (CC) protocols.
CC protocols handles the interleaving of concurrent operations by ensuring the transactional correctness as two essential properties: serializability (guarantee of consistent data snapshot)~\cite{Weikum2001TransactionalRecovery} and linearizability (guarantee of non-stale data snapshot)~\cite{Herlihy1990Linearizability:Objects}.
In theory, we can decide whether an omission of an update satisfies the correctness by finding \textit{version order}~\cite{Bernstein1983MultiversionAlgorithms} related to all update.
If a version order found and its correctness is verified by a protocol, we can skip locking, buffer updates, and persistent logging while preserving the transactional correctness.
However, to the best of our knowledge, no existing methods leverage the notion of version orders.
It is because the na\"ive approach requires the expensive acyclicity checking of all possible dependency graphs based on all transactions and all possible version orders, which has been proven as the NP-Complete problem~\cite{Weikum2001TransactionalRecovery}.

In this paper, we propose a versatile protocol extension method, \textbf{scheduling space expander (SSE)}.
The contributions of SSE are threefold.

    {\bf C1:~} SSE reduce the verification cost in a polynomial-time by testing only a single \textbf{erasing version order} which is generated by SSE's data structure.
    With erasing version order, the verification needs only for a single subgraph of concurrent transactions.
    SSE also keeps a version order and its testing algorithm by conventional protocols; if an erasing version order does not found or failed correctness testing, then SSE delegates the control flow to conventional protocols.
    Thus, SSE can omit updates but does not directly abort any transaction.
    It indicates that SSE purely expands the scheduling space of conventional protocols.

    {\bf C2:~} We developed \textbf{epoch-based SSE (ESSE)} to introduce optimizations for SSE. With the epoch framework, ESSE reduces the number of target transactions in the SSE's verification and encodes the footprints of these target transactions into a 64-bits data structure. As a result, a protocol expanded by ESSE can generate an erasing version order and execute correctness testing in a latch-free manner. If a transaction passes the testing, it can omit its write operations with a bit of atomic operations such as Compare-And-Swap.

    {\bf C3:~} We demonstrated that SSE and ESSE are applicable to various protocols. We applied ESSE to two state-of-the-art 1VCC and MVCC protocols (Silo~\cite{Tu2013SpeedyDatabases} and Cicada-based MVTO~\cite{Lim2017Cicada:Transactions}); then, we evaluated the performance on the TATP, YCSB~\cite{Cooper2010BenchmarkingYCSB}, and TPC-C benchmarks in a 144-core environment. Figure~\ref{fig:TATP-VaryingDIQueryRatio} shows that ESSE successfully mitigates the performance problem of data ingestion queries and improved the performance on the TATP benchmark. This is because ESSE appropriately omitted unnecessary versions, as illustrated in the experiment (Figure~\ref{fig:tatp}c in Section~\ref{sec:tatp}).
    In Table~\ref{tab:comparison} we present the difference between conventional protocols and our proposal.
    ESSE extends various state-of-the-art protocols such as Silo and MVTO to enable omitting write operations while presesrving the transactional correctness.

The rest of this paper is organized as follows. Section~\ref{sec:preliminaries} describes the preliminaries.
Section~\ref{sec:omittable_theo_and_ex} proposes the notion of safely omittable transactions and its correctness testing algorithm.
Section~\ref{sec:sse} presents the scheduling space expander (SSE) scheme, which can generate safely omittable transactions.
Section~\ref{sec:esse} presents ESSE, which is the  optimization technique for SSE based on the epoch framework.
Section~\ref{sec:evaluation} reports our evaluation of the proposed scheme. Finally, Section~\ref{sec:relatedwork} describes related work, and Section~\ref{sec:conclusion} concludes the paper.

\newcommand{\rulesep}{\unskip\ \vrule\ }
\begin{figure*}[t]
  \subfloat[$x_2 <_v x_1$: violates serializability]
  {\includegraphics[clip,width=0.23\textwidth]{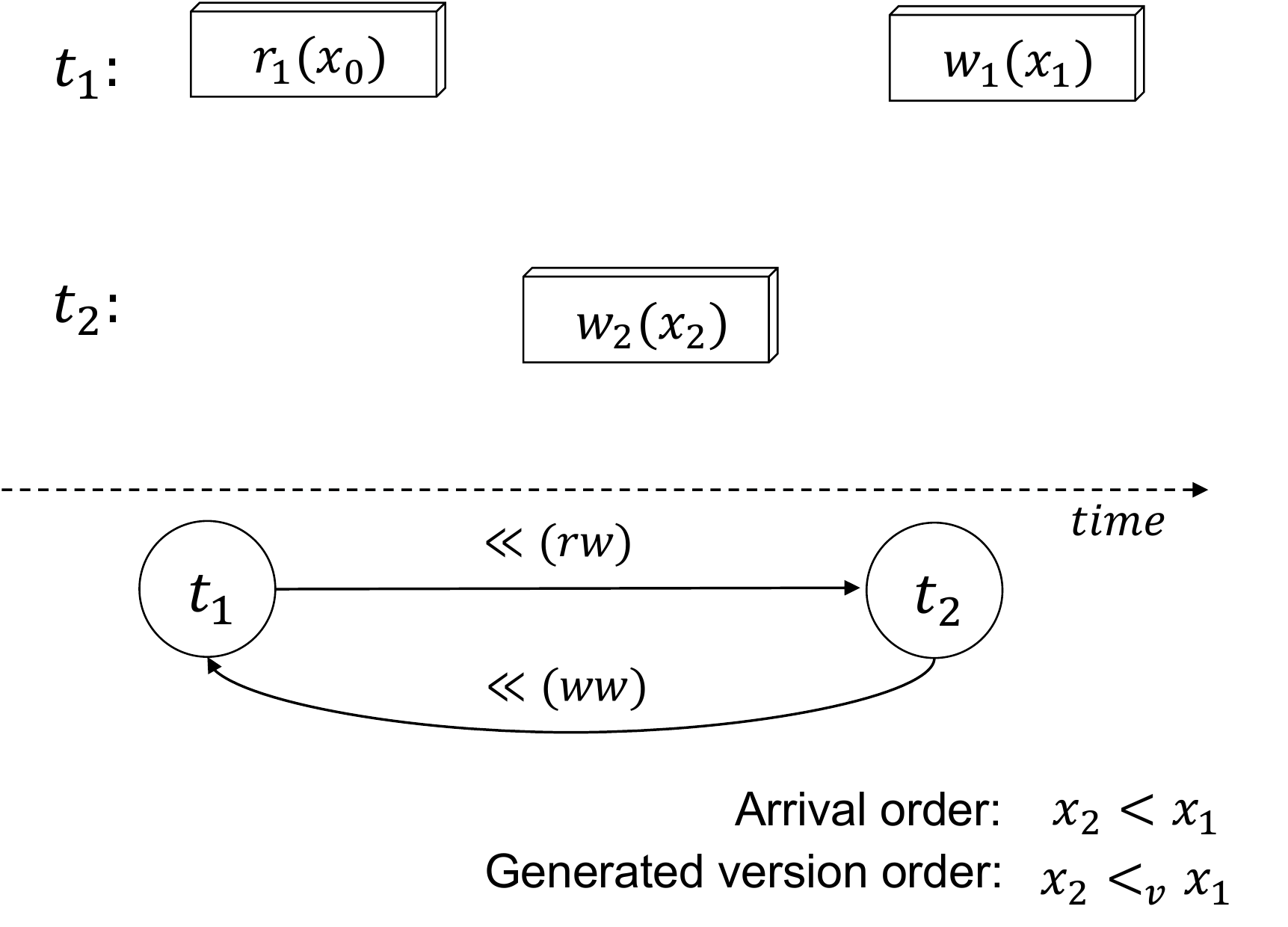}
  }
  \hfill
  \subfloat[$x_1 <_v x_2$: safely omittable]
  {\includegraphics[clip,width=0.23\textwidth]{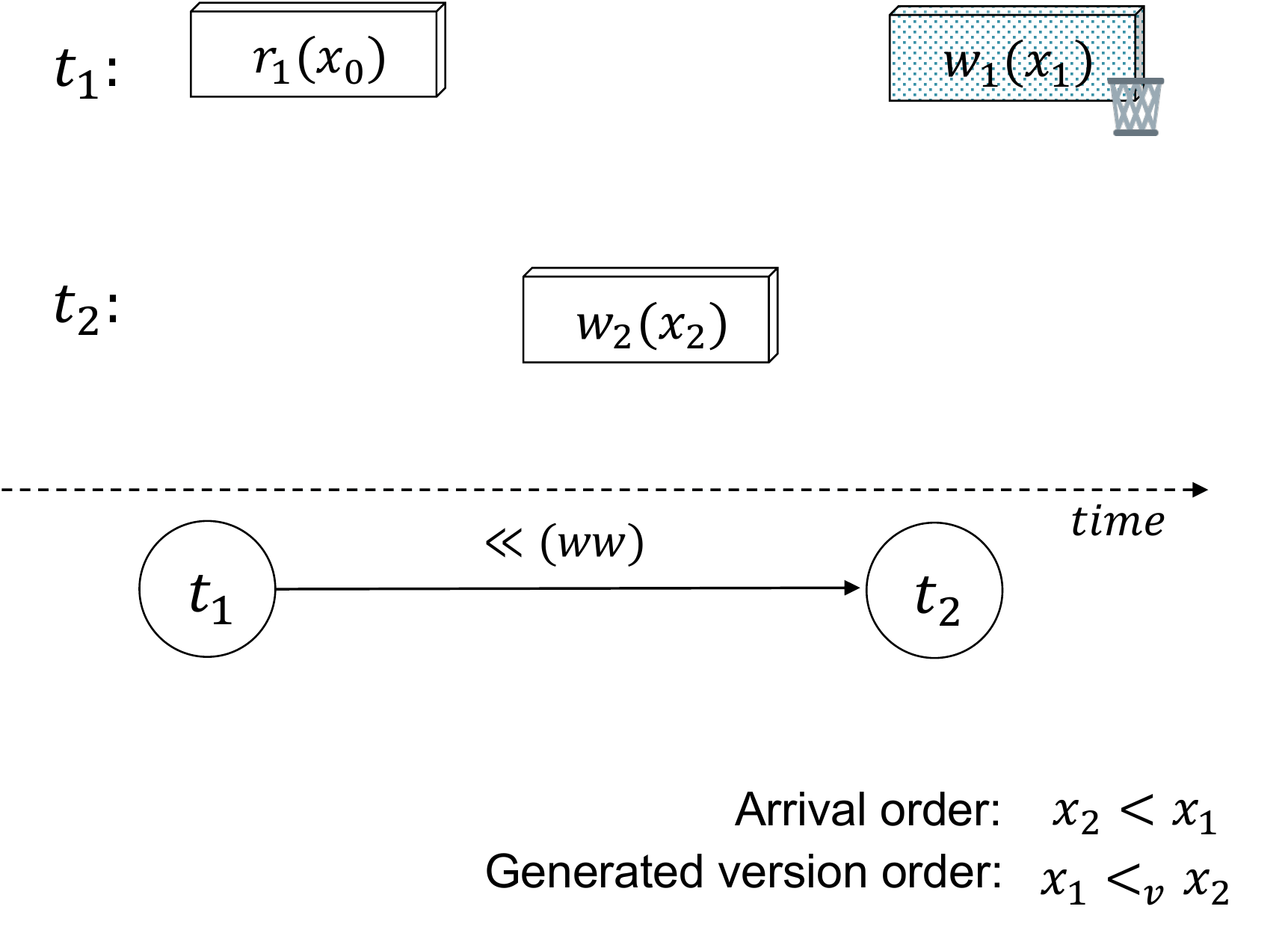}
  }
  \hfill
  \rulesep
  \hfill
  \subfloat[$x_2 <_v x_1$: violates linearizability ]
  {\includegraphics[clip,width=0.23\textwidth]{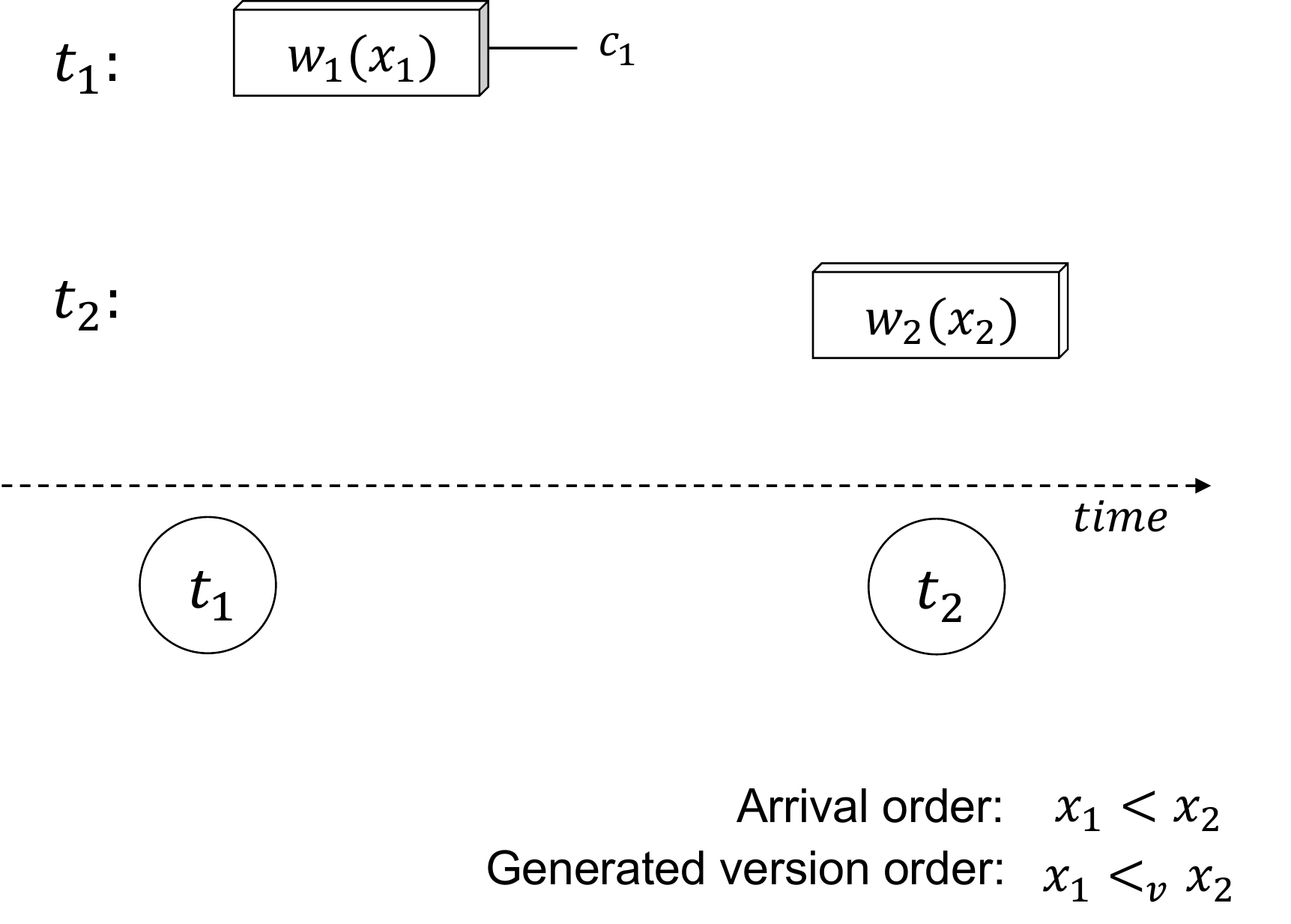}
  }
  \hfill
  \subfloat[$x_2 <_v x_1$: safely omittable]
  {\includegraphics[clip,width=0.23\textwidth]{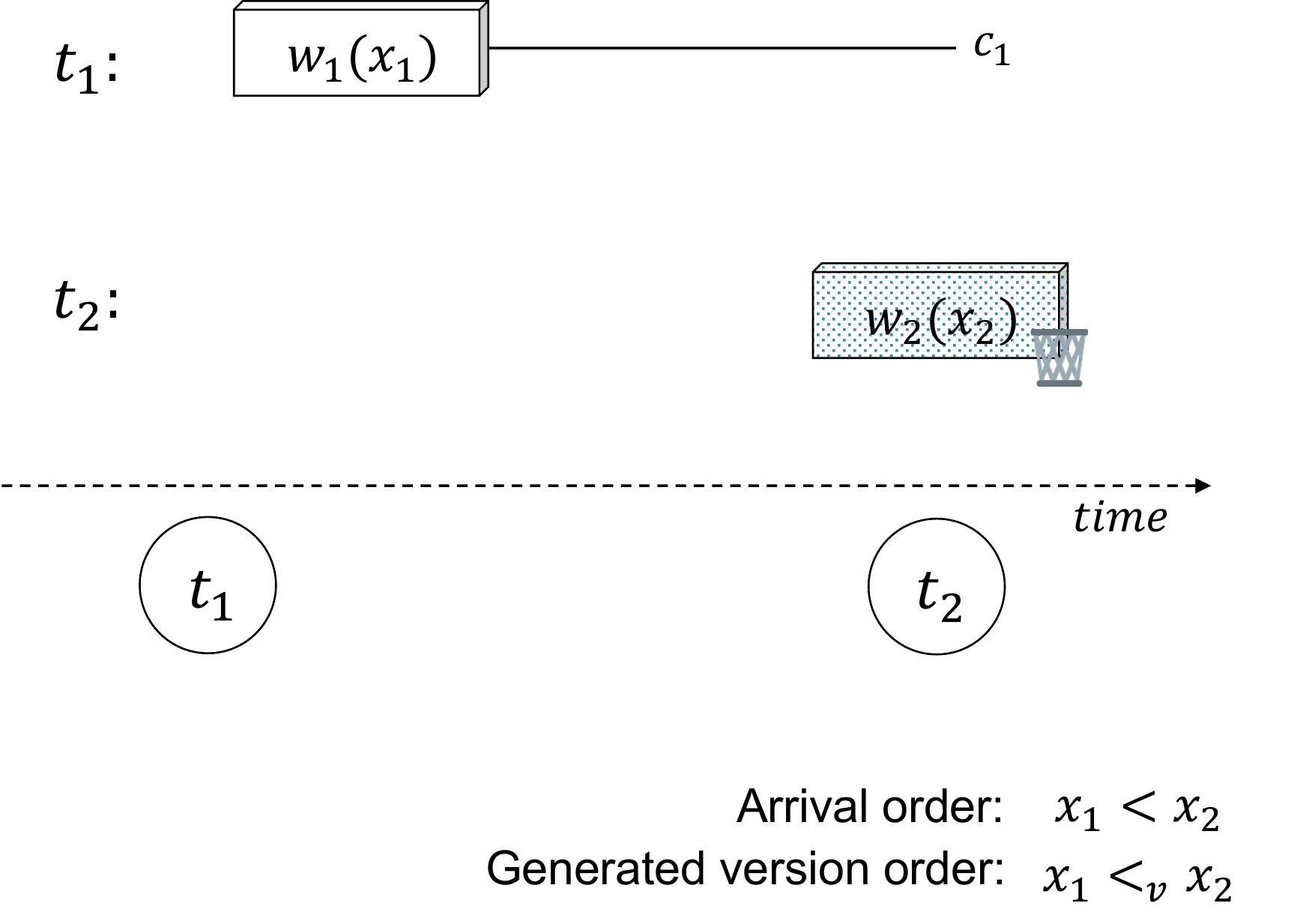}
  }

  \caption{Examples of schedules with version orders. The dotted operations depicted with wastebaskets omit its versions. Only cases (b) and (d) has safely omittable transactions (i.e., ensure the transactional correctness of unpublished transactions).}
  \label{fig:wrw_optim}
\end{figure*}

\section{Preliminaries}
\label{sec:preliminaries}

\begin{table}[t]
  \small
  \begin{tabular}{|c|l|} \hline
    Notation   & Definition                                       \\ \hline \hline
    $t_i$      & $i$-th transaction; an ordered set of operations \\
    $x_i$      & a version of data item $x$                       \\
    $w_i(x_i)$ & a write operation; $t_i$ writes $x_i$            \\
    $r_i(x_j)$ & a read operation; $t_i$ reads $x_j$              \\
    $c_i$ & a commit operation of $t_i$  \\
    $a_i$ & a abort operation of $t_i$  \\
    $rs_i$     & a set of versions read by $t_i$                  \\
    $ws_i$     & a set of versions written by $t_i$               \\
    \hline
  \end{tabular}
  \caption{Frequently used symbols and notations}
  \label{tab:symbols}
  \vspace{-14pt}
\end{table}

We mainly use the notations derived from Weikum et al.~\cite{Weikum2001TransactionalRecovery}. Table~\ref{tab:symbols} shows frequently used symbols and notations. Let $t_i, x_i, w_i(x_i)$, and $r_i(x_j)$ be $i$-th transaction, a version of data item $x$, a write operation, and a read operation, respectively.
$t_i$ has an ordered set of operations. $w_i(x_i)$ means $t_i$ writes $x_i$. $r_i(x_j)$ means $t_i$ reads $x_j$.
$ws_i$ and $rs_i$ represent the set of read and write operations of $t_i$, respectively.
$c_i$ and $a_i$ represents $t_i$'s termination operation; commits and aborts, respectively.

\subsection{Transactional Correctness}
We define the \textbf{transactional correctness} of our intended services as recoverability~\cite{Hadzilacos1988ASystems} and strict serializability~\cite{Herlihy1990Linearizability:Objects}.
We assume no transaction reads any value written by uncommitted transactions. This constraint ensures recoverability.
Strict serializability consists of serializability and linearizability~\cite{Herlihy1990Linearizability:Objects}. Serializability is necessary to provide consistent data snapshots for real-time operations of our intended IoT/Telecom applications, i.e., readings results without inconsistent or partial updates.
Among multiple notions of serializabilities, we use multiversion view serializability (MVSR) because this property provides the widest scheduling space~\cite{Weikum2001TransactionalRecovery}.
Linearizability refers to the wall-clock ordering constraints among non-concurrent transactions. If a database has the guarantee of linearizability, it ensures to prevent stale reads and writes, i.e., reading and writing of outdated versions.

\footnotetext[1]{In the original MVSG definition~\cite{Bernstein1983MultiversionAlgorithms}, all edges were denoted as $\to$, and the version orders for the schedule and for each data item had the same notation $\ll$; thus, when $x_i$ preceded $x_j$, it was denoted as ``$x_i \ll x_j$ in $\ll$.''}

Bernstein et al. proposed the multiversion serialization graph (MVSG)~\cite{Bernstein1983MultiversionAlgorithms} and proved that \textit{a schedule is MVSR if and only if there exists an acyclic MVSG}.
An MVSG has nodes for all committed transactions in the schedule.
The edges are added by a given schedule and a \textit{version order} for the schedule.
There are two types of version orders: version order for a data item and for a schedule.
A version order for a schedule is a union of all version orders for data items.
When $x_i$ precedes $x_j$ in a schedule, we denote $x_i <_v x_j$ and we refer it to as a version order for a data item.
With a version order for a schedule, the edges of MVSG are added for each triple of distinct operations $w_j(x_j)$, $r_i(x_j)$, and $w_k(x_k)$, where $t_i \ne t_k \ne t_j$.
There are three types of edges: (1) $t_j \WR{} t_i$ indicates that $t_j$ writes a version $x_j$ and $t_i$ reads it. (2) If $x_j <_v x_k$, then $t_i \VORW{} t_k$ indicates that $t_i$ reads a version $x_j$ and its version order precedes $t_k$'s version $x_k$. (3) Otherwise, $t_k \VOWW{} t_j$ indicates that $t_k$ writes a version $x_k$ and its version order precedes $t_j$'s version $x_j$.
Note that the original notation\footnotemark[1] of the MVSG does not include the dependency types for transaction orders.
Of course, Bernstein's definition has no problem.
However, we introduce the above notations to clarify our proof of correctness theorems.

\subsection{Data Ingestion Queries and Write Omission Technique}

The data ingestion queries are used to aggregate updates from sensors and mobile devices in IoT/Telecom applications.
This query has long been discussed in non-transactional systems such as streaming databases~\cite{DBLP:conf/vldb/TuLPY06,DBLP:journals/tmc/TahirYM18}.
The data ingestion query does not need to be transactional if it ingest data for historical analysis.
However, if the data is used in real-time operations that manipulate real-world actuators, we need to use CC protocols because such operations require the transactional correctness~\cite{meehan2017data,10.1145/2882903.2899406}.
Unfortunately, for write contended workloads such as data ingestion queries, it has been studied that existing state-of-the-art CC protocols do not exhibit high throughput~\cite{Yu2014StaringCores,Wu2017AnControl} since they need to use locking to satisfy serializability.

To solve the performance degradation problem on data ingestion queries, non-transactional streaming databases uses the write omission technique such as load shedding~\cite{DBLP:conf/vldb/TuLPY06}.
In transaction processing, such write omission technique is as known as the Thomas write rule (TWR)~\cite{Thomas1977ABases}, which is an optimization rule for the timestamp ordering (T/O) CC protocol.
With the TWR, a transaction can avoid installing a write of a data item $x$ when the transaction's timestamp is less than the $x$'s timestamp, which has already been installed.
However, it is unclear whether or not an omission satisfies the transactional correctness, and it is also unclear whether or not this rule can apply to other modern protocols.

\section{Safely Omittable Transactions}
\label{sec:omittable_theo_and_ex}

In this section, we introduce the definition of \textbf{safe omittable transactions} for any protocol to utilize the technique of write omission, and its validation algorithm with MVSG.

We first provide the following definition:

\begin{definition}[Unpublished]
    \label{def:unpublished}
    An \textbf{unpublished} transactions is a transaction which does not execute installing and logging of its write set into storage.
\end{definition}

\begin{definition}[Safely Omittable]
    \label{def:omittable}
    An unpublished transaction $t_j$ is \textbf{safely omittable} if $c_j$ does not affect the correctness.
\end{definition}

The key aspects of the Definition~\ref{def:omittable} is that the transaction can commit without publishing its write set.
It's versions must be unread by concurrent transactions and also future transactions.
Therefore, $t_j$ is unnecessary for other transactions; we can skip both buffer updates and persistent logging for safely omittable transactions.

To test whether a transaction $t_j$ is safely omittable, we need to verify the correctness.
It achieved by the notion of MVSG.
Specifically, we have to generate a version order for the schedule, and then test 1) the MVSG's acyclicity and 2) the wall-clock ordering among non-concurrent transactions.
Figure~\ref{fig:wrw_optim} illustrates these testing with two example schedules and four version orders.
Safely omittable transactions are grayed out and unpublished updates are marked with a trash box.
The pairs of (a)-(b) and (c)-(d) have the same schedule, but have different version orders and transaction lifetimes, respectively.
These difference result in (b) and (d) include safely omittable transactions, while (a) violates serializability and (c) violate linearizability.
Note that the operations arrives in order of wall-clock time depicted as left-to-right, but we draws MVSGs with version orders that are generated regardless of the arrival order.

\textbf{Serializability.} In (a) and (b), $t_1$ executes read-modify-write into $x$ (reads $x_0$ and writes $x_1$ as the next version) and $t_2$ executes blind write (writes $x_2$ as any version) over the same data item.
If we generate a version order $x_2 <_v x_1$ and omit $t_2$ as seen in (a), the edges of MVSG represent a cycle $t_1 \to t_2 \to t_1$.
This is because $t_1$ read-modify-writes to the just next version of $x_0$ and thus, any transaction can place a version as the middle of $x_0$ and $x_1$.
However, if we generate a version order $x_1 <_v x_2$ as seen in (b), $t_1$ is safely omittable since the MVSG has the acyclic form.

\textbf{Linearizability.} In (c) and (d), there exists only blind updates.
Therefore, both MVSGs are edgeless and acyclic.
However, if we generate improper version order and omit wrong versions, database violates linearizability.
In the case (c) we generate the version order $x_2 <_v x_1$.
This version order does not match the wall-clock ordering of transactions; $t_1$ and $t_2$ are non-concurrent transactions and thus the order of these transactions' versions must be $x_1 <_v x_2$.
Hereafter, the transactional correctness of database is lost by write omission.

We can see the rules and limitations for creating safely omittable transactions from the above examples.
Examples (a) and (b) indicate that there must be at least one \textit{blind update}~\cite{Weikum2001TransactionalRecovery}.
If there is no blind update but we omit an update, the correctness testing will not pass regardless of what version order we create.
In addition, examples (c) and (d) indicate that the blind update must be written by a concurrent transaction.
Hence, to create safely omittable versions, we have three questions that have never been comprehensively studied to the best of our knowledge:

\begin{quote}
\begin{enumerate}[Q1:]
    \item How to find concurrent blind updates to generate a version order?
    \item How to test the correctness?
\end{enumerate}
\end{quote}

In this paper, we package the solutions to these three problems into a single CC protocol extension called SSE.

\section{SSE: Scheduling space expander}
\label{sec:sse}

In this section, we propose the \textbf{scheduling space expander (SSE)} which adds another control flow to conventional protocols for generating safely omittable transactions.
SSE solve the problems shown in the previous section as followings:

\begin{quote}
\begin{enumerate}[\bfseries \textrm{A}1:]
    \item SSE selects and manages concurrent blind updates as the \textit{pivot versions} to generate an \textit{erasing version order} which assumes there exists safely omittable transactions~(Section~\ref{sec:generate_additional_version_order}).

    \item SSE uses MVSG to test the correctness of erasing version order efficiently~(Section~\ref{sec:validation}).
\end{enumerate}
\end{quote}

We first introduce an \textit{erasing version order}, which is a version order generated by SSE to reduce the computational cost of the correctness testing (Section~\ref{sec:generate_additional_version_order}).
We next outline how tests an erasing version order (Section~\ref{sec:validation}).
We then show the SSE's control flow to expand a protocol and improve its performance on data ingestion queries (Section~\ref{sec:control_flow}).

\subsection{SSE's Version Order Generation}
\label{sec:generate_additional_version_order}

\begin{figure}[t]
    \includegraphics[width=0.45\textwidth]{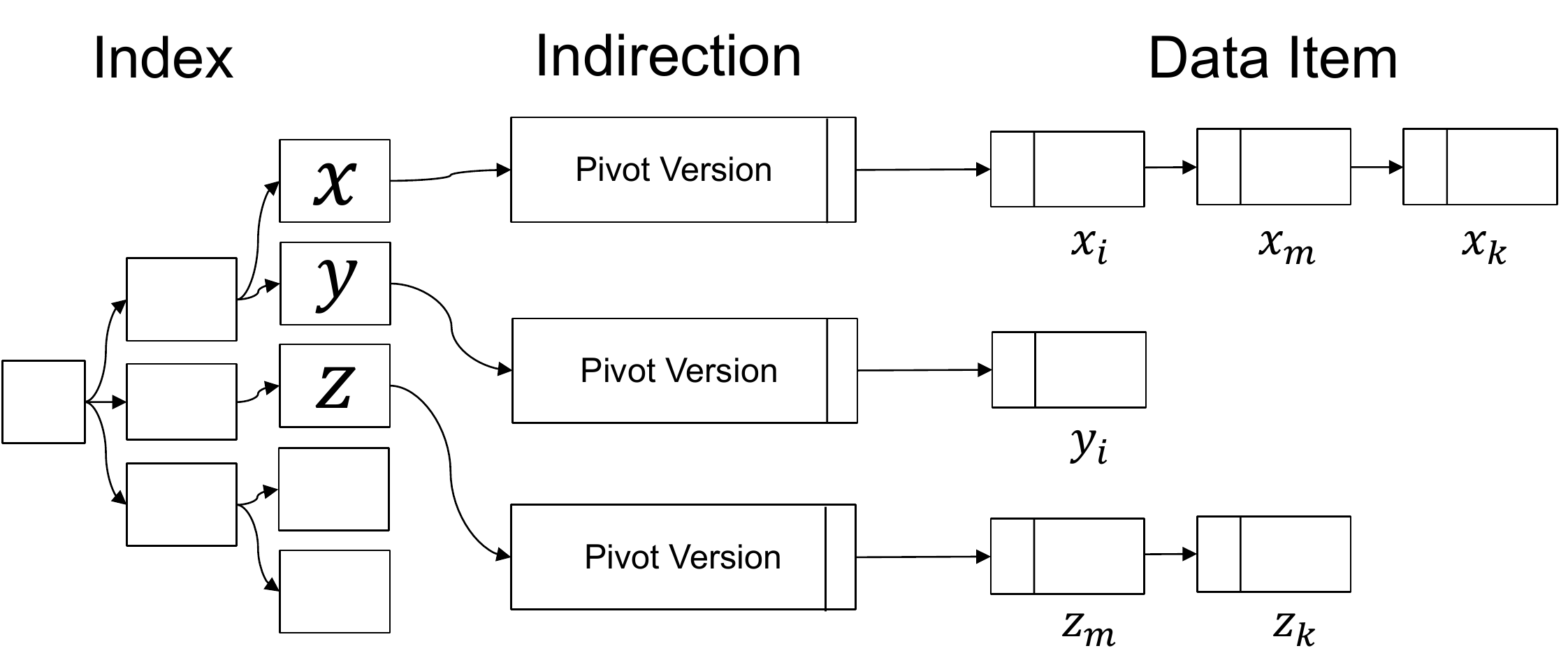}
    \caption{Overall structure of SSE implementation.}
    \label{fig:whatisfv}
\end{figure}

The serializability theory indicates a write is safely omittable if there exists a version order which draws acyclic MVSG.
However, it is impractical to try the testing with all possible version orders since the testing with all  version orders is proven as NP-complete~\cite{Weikum2001TransactionalRecovery,Bernstein1983MultiversionAlgorithms,Papadimitriou1986TheControl}.
To perform the test efficiently, SSE incorporates heuristic restrictions in generating candidate version orders.
When an active transaction $t_j$ arrives, SSE generates an \textbf{erasing version order}.
It is a version order which assumes that all writes in $t_j$ are safely omittable.
Formally, an erasing version order satisfies the following three conditions:
(1) SSE changes the version order only for data items that $t_j$ is updating.
From this condition, SSE's correctness testing can focus on the MVSG's subsets that include the node of $t_j$; if correctness is violated, it will be due to a change of a version order for a data item related to $ws_j$.
(2) Each version $x_j$ is the just before version of a blind update.
We add this condition to hold the unread condition of safely omittable versions; the non-latest versions become stale and not requested by subsequent transactions.
In addition, we enforce that the following version must be blind update.
As described at Figure~\ref{fig:wrw_optim}-(a) in Section~\ref{sec:omittable_theo_and_ex}, if $x_j$ is placed on the middle of read-modify-write, then MVSG always become cyclic.

As a concrete way to create an erasing version order, SSE selects a blind update as \textit{pivot version} for each data item.
A pivot version is a landmark for generating the erasing version order; it tells other transactions to ``place your version just before this pivot version''.
For example, in SSE, a transaction $t_j$ generates an erasing version order such that all its versions are located immediately before the pivot versions.
If there exists pivot versions $x_{pv}, y_{pv}, ...$, then $t_j$ creates an erasing version order $x_j <_v x_{pv}, y_j <_v y_{pv}$.
Figure~\ref{fig:whatisfv} shows the overall structure of our prototype implementation of the pivot versions.
We implemented the pivot versions by adding a single indirect reference for each data item.
We assume that database has a tree-like index, and that every data item is accessed from its leaf nodes.
In SSE, every index leaf node has a pointer to a pivot version, which is the indirection object to data item.
Each data item is represented as a singly linked list starting from the pivot version.
SSE completes the correctness testing of erasing version order only with pivot versions; a pivot version includes footprints of reachable transactions, as described in the later Section~\ref{sec:validation}.

\subsection{Correctness Testing}
\label{sec:validation}

With an erasing version order, SSE tests the transactional correctness efficiently.
To test serializability, it is sufficient to test only the MVSG paths that include a node of $t_j$.
To ensure linearizability, all of $t_j$'s reachable nodes that appear in the serializability test must be concurrent with $t_j$.
We define two types of node sets, \textit{successors}, and \textit{overwriters} (abbreviate as $s_j$ and $o_j$), in accordance with the outgoing edge from $t_j$ as follows:

\begin{definition}
    [Type of Reachable Transactions] For the transactions directly reachable from a transaction $t_j$, we define the following two sets:
    $$s_j := \{t_k| t_j \VOWW{} t_k\},\ \  o_j := \{t_k| t_j \VORW{} t_k\}$$
\end{definition}

The following theorem is derived from this definition:

\begin{theorem}
    [Directly Reachable Transactions]
    \label{theo:what_are_RN}
    If a schedule satisfies recoverability and an edge $t_j \to t_k$ exists, then the directly reachable transaction $t_k$ is in either $o_j$ or $s_j$.
\end{theorem}
\begin{proof}
    An MVSG has only three types of edges: $t_j \WR{} t_k, t_j \VORW{} t_k$, and $t_j \VOWW{} t_k$. Because recoverability is satisfied and $t_j$ is active, edge $t_j \WR{} t_k$ does not exist.
\end{proof}

From Theorem~\ref{theo:what_are_RN}, SSE can test correctness by testing all paths starting from these two sets of transactions. Therefore, SSE separates the procedure of correctness testing into two sub-testings as described below.

\textbf{Testing of successors.} No studies have focused on the testing of successors because conventional protocols always write the incoming $t_j$'s versions as the latest versions, so the set of successors is empty.
Algorithm~\ref{alg:successors_validation_pseudo} provides a testing procedure for $s_j$.
Let $rs_i$ be a set of versions read by $t_i$, and let $ws_i$ be a set of versions wrtten by $t_i$.
In step (A), it collects transactions that are included in or reachable from $s_j$. In step (B), for linearizability, it tests the concurrency between $t_j$ and each transaction $t_m$ in the collected transactions. In steps (C) and (D), the algorithm tests serializability. Because each $t_m$ is a transaction included in or reachable from $s_j$, a path $t_j \VOWW{} ... , \to t_m$ already exists. Therefore, if there is no path $t_m \to t_j$ for each $t_m$, the MVSG is acyclic; steps (C) and (D) thus focus on the last such edge $t_m \to t_j$. Note that transaction $t_m$ does not have a path $t_m \VOWW{} t_j$. This type of edge is added to the MVSG only if some committed transactions read some version in $ws_j$, and such a read operation is not permitted to enforce recoverability; because $t_j$ is an active transaction, no committed transaction can read version $x_j$.
Therefore, we only need to check the types of the last edges $t_m \WR{} t_j$ and $t_m \VORW{} t_j$. Accordingly, in step (C), the algorithm checks the last edges $t_m \WR{} t_j$. It checks whether $y_n$ in $rs_j$ is the same with $y_m$ in $ws_m$. This is because, if this condition holds, then there exists a cyclic path $t_j \VOWW{} ... \to t_m \WR{} t_j$. Similarly, in step (D), the algorithm checks the last edges $t_m \VORW{} t_j$ by confirming that there exists $y_j$ in $ws_j$ that is newer than $y_g$ in $rs_m$. This is because there exists a cyclic path $t_j \VOWW{} ... \to t_m \VORW{} t_j$ if the condition holds. Consequently, if the testing of steps (C) and (D) are passed, then no transaction in $s_j$ can reach $t_j$.

\begin{algorithm}[t]
    \small
    \KwIn{$t_j$}
    \KwOut{whether or not $c_j$ keeps serializability}
    
        $T := \{t_k, t_i | t_k \in s_j \land t_i$ is reachable from $t_k$  in MVSG $\}$
        \tcp*[r]{(A)}

        \ForAll{$t_m$ in $T$}{

            \If
                (\tcp*[f]{(B)})
                {$t_m$ commits before $t_j$'s beginning}{
                    \Return{strict serializability is not satisfied}
            }
            \ForAll{$y_m$ in $ws_m$}{
                \ForAll{$y_n$ in $rs_j$}{
                    \If
                        (\tcp*[f]{(C) $\WR{}$ to $t_j$})
                        {$y_m = y_n$}{
                            \Return{MVSG is cyclic}
                        }
                }
            }
            \ForAll{$y_g$ in $rs_m$}{
                \ForAll{$y_j$ in $ws_j$}{
                    \If
                        (\tcp*[f]{(D) $\VORW{}$ to $t_j$})
                        {$y_g <_v y_j$}{
                            \Return{MVSG is cyclic
                        }
                    }
                }
            }
        }

        \Return{MVSG is acyclic}
    \caption{Correctness testing for $s_j$}
    \label{alg:successors_validation_pseudo}
\end{algorithm}

\textbf{Testing of overwriters.} The detailed algorithms of this subsets are beyond the scope of this paper since we can use existing algorithms, such as anti-dependency validation, from conventional protocols.
For example, Silo~\cite{Tu2013SpeedyDatabases} has one of the simplest approaches: it checks $o_j = \phi$ by testing whether each version in $rs_j$ has been overwritten and a newer version exists.
If both testings of successors and overwriters are succeeded, we can commit $t_j$ without any correctness violation since $t_j$'s all write operations are safely omittable.

An important fact is that the set of transactions $t_{pv}$ that wrote the pivot versions, is equivalent to the successors, $s_j$.
This is because the pivot versions are the just next versions for $t_j$’s updates.
Therefore, we can implement the Algorithm~\ref{alg:successors_validation_pseudo} by using the read/write sets of the transactions $t_m$, that includes $t_{pv}$ and reachable transactions from some $t_{pv}$.
To this end, each pivot version hold the footprints of transactions that have read or written a version greater than or equal to this pivot version.

\subsection{Control Flow}
\label{sec:control_flow}

\begin{figure}[t]
  \includegraphics[width=0.45\textwidth]{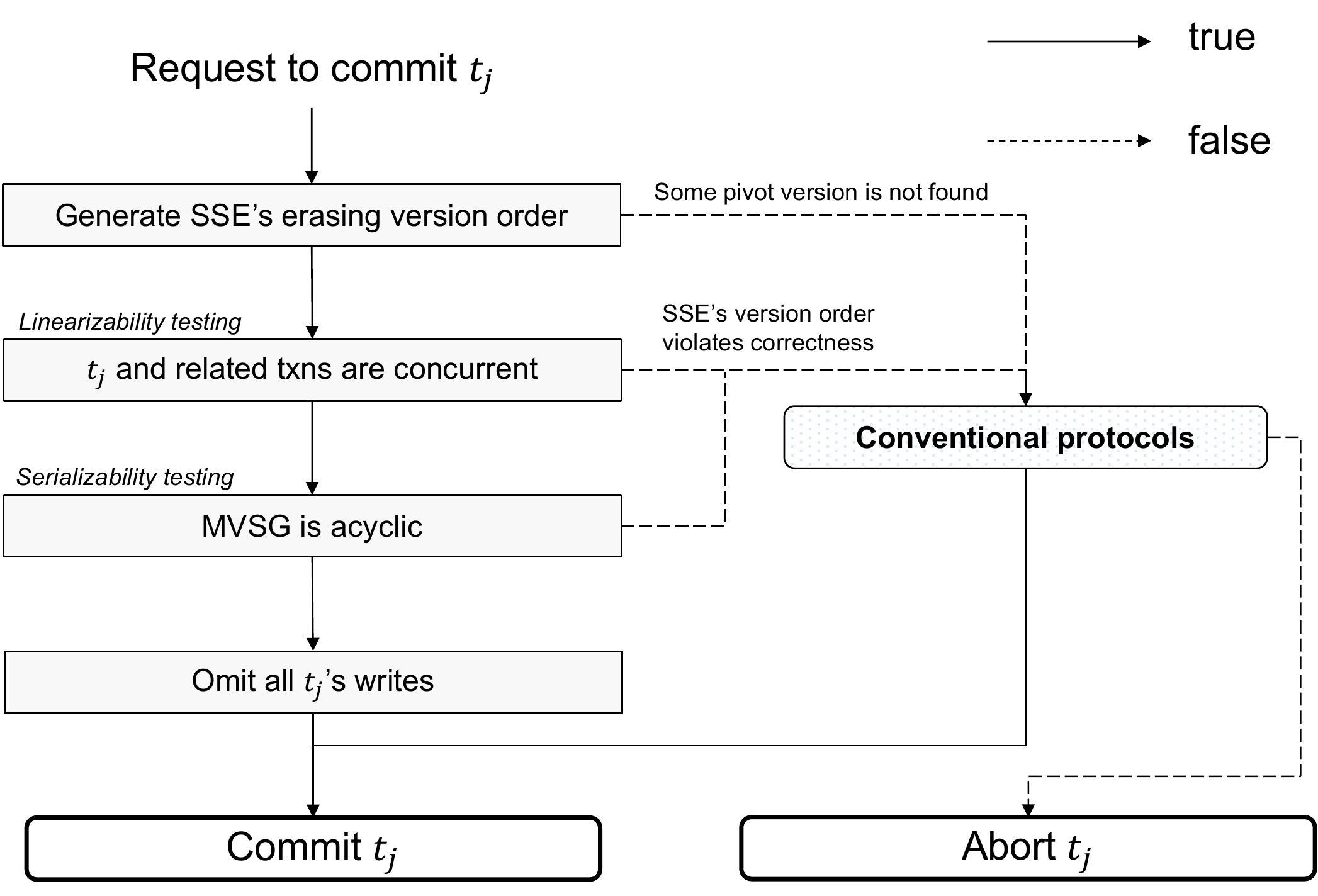}
  \caption{Control flow of SSE with conventional protocols.}
  \label{fig:flowchart}
  \vspace{-10pt}
\end{figure}

We revealed that an erasing version order helps to confirm that a write is safely omittable, and we also revealed the algorithm to test it.
The next challenge is how to apply this version order and testing algorithm to conventional CC protocols.
It is desirable for us not to change the performance characteristics and the specifications of conventional CC protocols drastically.
However, if we simply change a conventional protocol to generate an erasing version order, the changed protocol will be useless; since erasing version order always try to write non-latest versions, we would not successfully update the latest versions.
Even if contended updates occur into the same data item, practical applications need to update the latest version periodically.
Rather than such changes, we propose \textbf{SSE}, an extension method that purely extends the scheduling space of conventional protocols.
SSE is not a protocol, but an extension of the protocol; instead of changing the version order and correctness testing algorithms of conventional protocols, SSE adds erasing version order and tests its correctness independently.

Figure~\ref{fig:flowchart} shows the control flow of an extended protocol.
SSE starts its processing upon the commit request of an active transaction $t_j$.
In the first step, SSE generates an erasing version order as described in Section~\ref{sec:generate_additional_version_order}.
If any version order is not be found, SSE delegates the processing of $t_j$ to a conventional protocol.
If an erasing version order is found, SSE then checks correctness.
For linearizability, it checks concurrency between $t_j$ and all related transactions.
For serializability, it checks whether MVSG has no cycle.
If both tests are passed, SSE omits all of $t_j$'s versions and commits $t_j$.
Otherwise, it delegates the processing of $t_j$ to a conventional protocol.
Note that SSE does not change the version order generated by the conventional protocol, and SSE can commit $t_j$ but does not directly abort $t_j$.
Even if SSE cannot commit $t_j$ with its erasing version order, it may be possible to commit $t_j$ with a version order from the conventional protocol.
It indicates that SSE purely expands the scheduling space of conventional protocols.

In the best case, we can omit updates of $t_j$ only by accessing the pivot versions since these indirection data structures have all the necessary data for correctness testing.
Otherwise, if an update $w_j(x_j)$ does not find a pivot version of $x$ or some testing fails, then we delegate the control and quit to execute testing of SSE.
In this case, pivot versions have almost negligible overhead for conventional protocols; it adds only a single indirection references for each data item.

\section{ESSE: Epoch-Based SSE}
\label{sec:esse}

In this section we introduce the optimization method of SSE, called  \textbf{epoch-based SSE (ESSE)}.
We first explain that the na\"ive implementation of SSE has performance problems in the correctness testing of successors.
As described in Section~\ref{sec:omittable_theo_and_ex}, this testing  requires both a read set ($rs_m$) and a write set ($ws_m$) for all reachable transactions $t_m$; it leads a huge overhead. Specifically, in step (A) of Algorithm~\ref{alg:successors_validation_pseudo}, the memory consumption from the footprints $rs_m$ and $ws_m$ increases as the number of testing target transactions $t_m$ increases. In addition, steps (C) and (D) have the synchronization problem: footprints $rs_m$ and $ws_m$ must be a concurrent data structure to commit transactions in parallel.

To solve these two performance problems, we propose ESSE, which is an optimized SSE implementation using epoch-based group commits~\cite{Tu2013SpeedyDatabases,Chandramouli2018FASTER:Updates}. An epoch-based group commit divides the wall-clock time into \textit{epochs} and assigns each transaction to an epoch. The commit operations for the transactions in an epoch are simultaneously delayed. Because they have the same commit point, they can be regarded as concurrent, whereas transactions in different epochs are not concurrent.
ESSE heavily utilizes this nature of concurrency in epochs.
Specifically, ESSE selects pivot versions as the first blind updates of data items for each epoch.
Then the linearizability testing become simple; we can test it easily by checking that the pivot versions are written in the same epoch with $t_j$'s one.
This tight coupling of epochs and pivot versions helps to mitigate the former performance problem of SSE, the huge memory consumption.
There is no need to retain the footprints of the transactions in old epochs because SSE fails the linearizability testing before the serializability testing for these transactions.

ESSE also includes two optimizations, as follows. (1) It creates an optimization rule by using epochs to reduce the number of testing target transactions $t_m$ (Section~\ref{sec:optimization}). (2) It provides a \textbf{reachability flag} for testing of successors efficiently; it is a latch-free implementation of SSE's pivot versions to solve the latter performance problem of SSE, the necessity of synchronization (Section~\ref{sec:implementation}).
After these explanation of ESSE, we show the actual application of ESSE to conventional protocols, and explain what kind of protocols are preferable for ESSE (Section~\ref{sec:protocol_comparison}).

\subsection{Optimization Rule}
\label{sec:optimization}

\begin{figure}[t]
    \includegraphics[width=0.5\textwidth]{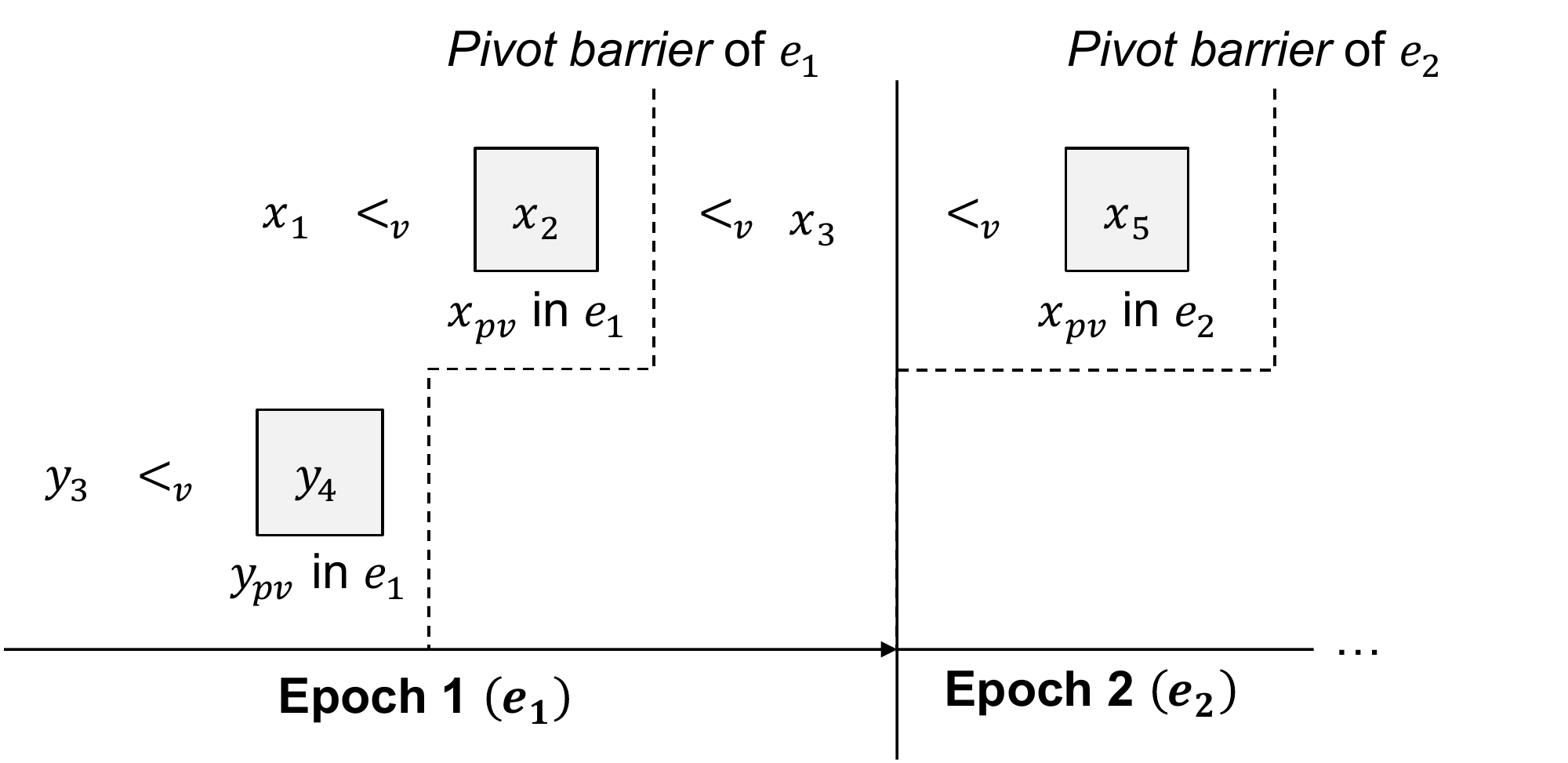}
    \caption{Pivot versions and pivot barriers of two epochs. All versions are not omittable and installed in memory. An epoch's pivot barrier consists of its pivot versions. The gray boxes depict the pivot versions, and the dashed lines depict the pivot barriers of each epoch. $t_3$ fails SSE's correctness testing and installs versions for both the right and left side of the pivot barrier. Such a transaction causes an MVSG cycle for subsequent transactions.}

    \label{fig:epoch_barrier}
\end{figure}


To reduce the number of testing target transactions, we add the optimization rules by using ESSE's pivot versions. For an epoch, we define a set of pivot versions as a \textbf{pivot barrier}. Using the pivot barrier, we obtain the following theorem, which provides important properties for reducing the target transactions.

\begin{theorem}
    [Pivot Barrier Violation]
    \label{theo:epoch-barrier}
    Let $t_j$ be an active transaction in epoch $e_1$, attempting to commit on ESSE.
    Let $et_1$ be the transactions that wrote the pivot barrier (pivot versions in the epoch) of $e_1$.
    The correctness testing of $successors_j$ passes if $t_j$ is not reachable from any transaction in $et_1$.
\end{theorem}
\begin{proof}
    To prove by contradiction, we first assume that no transaction exists that satisfies both conditions of Theorem~\ref{theo:epoch-barrier} and that the testing of $successors_j$ detects an MVSG cycle. Because the MVSG has a cycle, there exists a transaction $t_i$ from which $t_j$ is directly reachable, which satisfies the direct reachability condition in Theorem~\ref{theo:epoch-barrier}. By the assumption, $t_i$ does not satisfy the right side of the pivot barrier condition; thus, $t_i$ is not reachable from any $t_{pv}$. Therefore, the MVSG cycle is completed by the left side of the pivot barrier in the form $t_j \to ..., \to t_i \to t_j \to t_{pv}$. We focus on the first edge outgoing from $t_j$. Because all versions in $ws_j$ are immediate predecessors of the pivot versions, no transaction can write a version that is larger than a version belonging to $t_j$ but smaller than the pivot version. Therefore, to form the cycle on the left side of the pivot barrier, the type of edge in the cycle outgoing from $t_j$ must not be $\VOWW{}$.
    In addition, the type of edge must not be $\WR{}$ since $t_j$ is an active transaction and it is attempting to omit its writes.
    Therefore, the edge type must be $\VORW{}$.
    Then, from Theorem~\ref{theo:what_are_RN} and the assumption, $t_i$ is not in $successors_j$ and its reachable transactions.
    Therefore, $t_i$ is not a target transaction of the correctness testing of $successors_j$, and $t_j$ cannot detect an MVSG cycle.
\end{proof}

Theorem~\ref{theo:epoch-barrier} clarifies the type of transactions introducing MVSG cycles during the testing of successors. Because an MVSG cycle always includes the transactions that read or write a version on the right side of the pivot barrier and that $t_j$ is reachable from, ESSE only needs to store the footprints of such transactions.
If $t_j$ detects the existence of such transactions in the saved footprints, ESSE does not commit $t_j$; on the contrary, the existence of other transactions does not matter for $t_j$.
Therefore, we only need to manage the footprints of transactions that satisfy the two conditions in Theorem~\ref{theo:epoch-barrier}.
Specifically, we store the footprints of transactions that satisfy the first condition (right side of pivot barrier) since we cannot test whether a transaction satisfies the second condition without incoming transaction $t_j$.
Transactions that satisfy the first condition assumes that they are reachable to incoming $t_j$ and thus they store their read/write sets into pivot version.

Figure~\ref{fig:epoch_barrier} shows examples of pivot versions and their version order in each epoch. All versions are installed in shared memory, and all transactions $t_1, t_2, t_3, t_4$, and $t_5$ are committed. Let $t_j$ be an active transaction that belongs to epoch $e_1$ and requests to write $x_j$ and $y_j$. For $t_j$'s correctness testing of successors, the existence of $t_3$ is essential: because $t_3$ satisfies both the right-side and directly reachable conditions in Theorem~\ref{theo:epoch-barrier}, it generates a cycle in the MVSG and violates serializability in committing $t_j$ with ESSE's erasing version order. Therefore, before committing $t_j$, ESSE has to detect such an MVSG cycle from the footprint of $t_3$.
In contrast, because $t_1$ does not satisfy the right-side condition, by Theorem \ref{theo:epoch-barrier}, it is unnecessary to save its footprint.
In addition, by using epochs, ESSE reduces the number of necessary footprints and enables periodic garbage collection.
There is no need to manage the pivot versions and pivot barriers of old epochs.
Specifically, by linearizability, it is unnecessary to manage the footprints of any transactions except $t_5$ after all of $e_1$'s transactions are terminated; after $e_1$ finishes, we do not start any transactions belonging to $e_1$, and it is thus unnecessary to manage the pivot versions and footprints of $e_1$.

\subsection{Latch-free Implementation}
\label{sec:implementation}

\begin{figure}[t]
    \centering
    \includegraphics[width=0.40\textwidth]{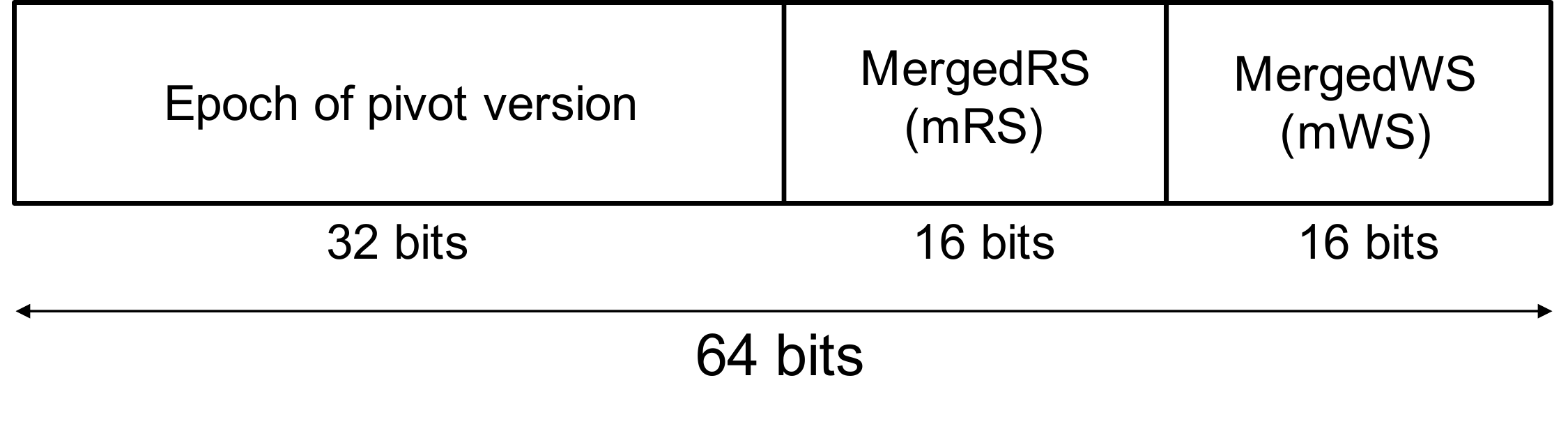}
    \caption{Layout of a reachability tracker}
    \label{fig:datastructure}
\end{figure}

\begin{figure}[t]
    \centering
    \includegraphics[width=0.37\textwidth]{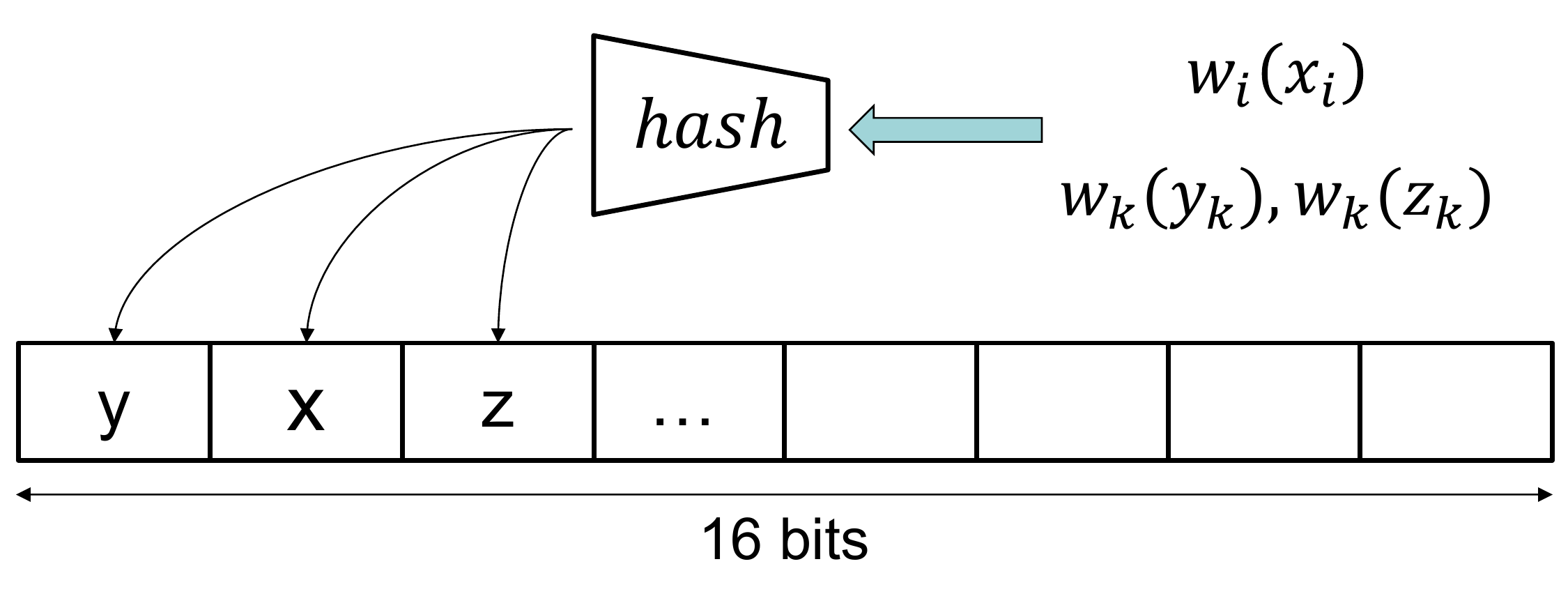}
    \caption{An example of mergedWS}
    \label{fig:merged_sets}
\end{figure}

\begin{algorithm}[t]
    \small
        \KwIn{$t_j$}
        \KwOut{whether or not $c_j$ keeps serializability}

        \ForAll{$x_j$ in $ws_j$}{
            {RT} := get\_reachability\_tracker\_of($x$)
            \tcp*[r]{(A)}

            \If{{RT}.epoch $\ne$ $t_j$'s epoch}{
                \Return{strict serializability is not satisfied}
                \tcp*[r]{(B)}
            }

            \If
                (\tcp*[f]{(C)})
                {any key in $rs_j$ exists in {RT}.mWS}{
                    \Return{MVSG may not be acyclic}
            }

            \If
                (\tcp*[f]{(D)})
                {any key in $ws_j$ exists in {RT}.mRS}{
                    \Return{MVSG may not be acyclic}
                }
        }
        \Return{strict serializability is satisfied}

    \caption{Correctness testing for $successors_j$ by using reachability trackers}
    \label{alg:successors_validation_detail}
\end{algorithm}

In ESSE, a pivot version for each data item is represented as a \textbf{reachability tracker}, which is a 64-bit data structure.
It is the detailed implementation of the pivot version as described in Figure~\ref{fig:whatisfv}.
A reachability tracker consists of an epoch and two fields (MergedRS and MergedWS). Figure~\ref{fig:datastructure} shows the reachability tracker's layout.

\begin{itemize}
    \item \textbf{Epoch.} This field stores the epoch of the pivot version as a 32-bit integer. 
    Each transaction fetches and assigns its own epoch from the global epoch at the beginning of the transaction. When a transaction executes a blind update and it becomes a pivot version, ESSE stores the transaction's epoch to this field.
    This field is used for the correctness testing.
    Linearizability is satisfied if the epochs of all collected reachability trackers have the same epoch as the active transaction's epoch.
    \item \textbf{MergedRS (mRS) and MergedWS (mWS).} These fields represent 16-bit bloom filters. Each filter stores all footprints of the transactions that satisfy the two conditions in Theorem~\ref{theo:epoch-barrier}; it stores the union of the keys of the read or write set for all transactions that read or write the greater or equal versions of pivot versions.
    Figure~\ref{fig:merged_sets} illustrates an example layout of the mWS.
    Each data item $x,y,z$ is mapped into the corresponding slot by a hash function $h$. ESSE uses this structure to test the serializability of successors.
\end{itemize}

The reachability tracker encapsulates the functionalities of the SSE's pivot version into the 64-bit data structure; it has a scheme for transactions' concurrency detector and footprints for reachable transactions.
Through this compression, ESSE accesses the reachability trackers in a latch-free manner by using atomic operations such as a compare-and-swap (CAS).

Algorithm~\ref{alg:successors_validation_detail} is a detailed implementation of Algorithm~\ref{alg:successors_validation_pseudo} that is based on using reachability trackers.
All steps (A-D) are equivalent in the two algorithms.
In step (A), it collects the reachability trackers for the data items in the $ws_j$ in order to generate erasing version order. In step (B), the algorithm checks linearizability by testing whether $t_j$ and the reachability trackers have the same epoch.
If it detects different epochs, the correctness testing for linearizability fails.
In steps (C) and (D), the algorithm checks serializability by testing whether $t_j$ has an incoming edge.
The reachability tracker exploits the bloom filter and thus it may produce false positives though it produces no false negatives.
When a hash function assigns different keys to the same slot, it may return that a non-existent cycle exists.

\begin{algorithm}[t]
    \small
    \KwIn{$t_j$}
    \KwOut{commit or abort}

    $t_j$.omittable := true

    \textit{\# Begin atomic section}

    \If{no write of $t_j$ is a blind update}{
        $t_j$.omittable := false
    }
    \If{correctness testing fails}{
        $t_j$.omittable $:=$ false \tcp*{(1)}
    }
    \If{$t_j$.omittable == false}{
        \textit{\# Try the conventional protocol}

        \If{the conventional protocol aborts $t_j$}{
            \Return{abort}
        }
    }
    \If
        (\tcp*[f]{(2)})
        {$t_j$ satisfies the first conditions of Theorem \ref{theo:epoch-barrier}}
    {
        \ForAll{$x_i$ in $rs_j$}{
            {RT} := get\_reachability\_tracker\_of($x$) \\
            {RT}.mRS.merge($rs_j$.keys) \\
            {RT}.mWS.merge($ws_j$.keys) \\
            \tcp*[r]{(3)}

        }
        \ForAll{$x_j$ in $ws_j$}{
            {RT} := get\_reachability\_tracker\_of($x$) \\
            \If{{RT}.epoch $\ne$ $t_j$.epoch}{
                \If
                    (\tcp*[f]{(4)})
                    {$w_j(x_j)$ is blind update}{
                        {RT}.epoch := $t_j$.epoch \\
                        init {RT}.mRS with $rs_j$ \\
                        init {RT}.mWS with $ws_j$ \\
                    }
            }\Else{
                {RT}.mRS.set($rs_j$.keys) \\
                {RT}.mWS.set($ws_j$.keys) \\
                \tcp*{(5)}
            }
        }
    }
    \textit{\# End atomic section}

    \If{$t_j$.omittable}{
        clear($ws_j$) \\
    }
    \Return{commit}
    \caption{Commit protocol of ESSE}
    \label{alg:NWR_updating_ds}
\end{algorithm}

\begin{table*}[hbt!]
    \centering
    \begin{tabular}{ccccc}
        \toprule
        Protocol          & Epoch-based group commit & Read/Write set & Testing method of overwriters & Optimistic CC \\
        \midrule
        Silo OCC          & Yes                      & Yes            & Yes                           & Yes                            \\
        Cicada MVTO       & No                       & Yes            & No                            & Yes                            \\
        2PL & No                       & No             & No                            & No                             \\
        \midrule
        Required for ESSE & Yes & Yes & Yes & Preferable \\
        \bottomrule
    \end{tabular}
    \caption{A list of components required by ESSE, and the comparison of protocols.
    In order to extend a protocol with ESSE, we need to add lacking components.
    Although OCC is not a requirement, it is a preferable property for ESSE in terms of performance.}
    \label{tab:extension}
\end{table*}

\footnotetext[2]{Here, the ``commit protocol'' refers to the processing at the time when a user application does not add any operation into the transaction. It does not refer to the ``commit phase'' in optimistic concurrency control (OCC).}

To use reachability trackers for correctness testing, we need to update them to maintain the MergedRS/WS of reachable transactions from the pivot version.
Thus, we add this updating before the commit of each transaction. Algorithm~\ref{alg:NWR_updating_ds} describes the ESSE commit protocol\footnotemark[2] of an active transaction $t_j$. In step (1), it tests the strict serializability of the $o_j$ and $s_j$ sets.
If the testing fails, the algorithm quits processing $c_j$ by using ESSE's version order and delegates subsequent processing to conventional protocols.
In step (2), if $t_j$ does not satisfy the conditions of Theorem~\ref{theo:epoch-barrier}, ESSE does not store a footprint.
Otherwise, ESSE stores $t_j$'s footprint in the reachability trackers.
In step (3), it stores all edges from the pivot barrier to $t_j$ to the mergedRS of the data item.
Specifically, the algorithm sets a bit flag in mergedRS and mergedWS for each data item in $rs_j$ and $ws_j$, respectively.
It also sets bit flags for all bits in the reachability trackers of $rs_j$.
In step (4), it updates the mergedWS.
If $x_j$ is the first blind update of the epoch, the algorithm resets the pivot version of data item $x$ to $x_j$. Otherwise, it adds $t_j$'s footprint and all edges from the versions in the pivot barrier to the mergedWS of the data item.
Note that steps (3) and (5) merge the bits in mRS/mWS of reachability trackers accessed by $t_j$ to keep reachability from the pivot versions; when $t_j$ accesses $x$ and $y$ and commits, there exist paths from the pivot versions of $x$ and $y$ to $t_j$ respectively, and thus reachability trackers of both data items should store these paths by merging.


Throughout the algorithm, we access the reachability trackers atomically in a latch-free manner by using the 64-bit data layout.
ESSE copies all reachability trackers to another location at the beginning of the commit protocol.
After testing and modification of mRS/mWS, it performs the CAS operation for all locations to update all fields atomically.
If a CAS operation fails, ESSE retries the commit protocol from the beginning.
If the bit arrays in the reachability tracker match exactly before and after the modifications, we can guarantee atomicity in a lightweight way via load instead of CAS to verify that no changes occurred concurrently.

\subsection{Extension Details}
\label{sec:protocol_comparison}

To apply ESSE to a conventional protocol, we need to add some components: epoch-based group commit, a read/write set of each transaction, reachability tracker, and the correctness testing of overwriters.
Table~\ref{tab:extension} summarizes the components required for the ESSE extension to a protocol.
If the protocol already has the necessary components, ESSE can use them straightforwardly.
For example, Silo has everything ESSE needs, and thus it is one of the preferable protocols.
Note that Table~\ref{tab:extension} does not include reachability tracker, but ESSE adds it for all protocols as described in Figure~\ref{fig:whatisfv}.
In contrast with Silo, the traditional two-phase locking protocol (2PL)~\cite{Gray1992TransactionTechniques} has none of the necessary components.
In addition, 2PL cannot utilize the version omission technique of ESSE since it writes a version to the data item immediately before the commit of transactions.
Moreover, if 2PL is used with some logging algorithms such as ARIES~\cite{Mohan1992ARIES:Logging}, it persists the log immediately.
The protocol of ESSE starts when the uncommitted versions are already installed and persist.
In this case, if the erasing version order is accepted, 2PL+ESSE must undo the uncommitted installed versions before unlocking.
Therefore, OCC is a desirable property for ESSE in terms of performance.

We extended two optimistic protocols: Silo and MVTO to Silo+ESSE and MVTO+ESSE, respectively.
We chose these protocols since they are modern fast protocols of the 1VCC and MVCC types, respectively.
In addition to 2PL, we did not include T/O with TWR for the experiments.
T/O inherently requires a centralized counter to generate monotonically increasing timestamps.
It is known that protocols with such a counter incur serious performance degradation in a many-core environment~\cite{Yu2014StaringCores}.

\textbf{Silo and Silo+ESSE.} Silo~\cite{Tu2013SpeedyDatabases} is an optimistic protocol that obtains state-of-the-art performance on read-intensive workloads.
When an application requests a transaction to commit, Silo acquires locks for all data items in the transaction's write set before installing new versions.
We ported Silo's correctness testing of overwriters to Silo+ESSE.
Silo's testing checks $o_j = \phi$; if a data item in $rs_j$ is overwritten, Silo and Silo+ESSE abort $t_j$.
Note that the original Silo stores records in the leaf nodes of the index directly.
Our implementation of Silo and Silo+ESSE may cause an overhead of a cache miss due to ESSE's reachability tracker indirection.

\textbf{MVTO and MVTO+ESSE.} MVTO~\cite{Bernstein1987ConcurrencySystems,Lim2017Cicada:Transactions} is a timestamp-based CC protocol that has multiversion storage.
We implemented MVTO on the basis of Cicada \cite{Lim2017Cicada:Transactions}, the state-of-the-art MVTO protocol.
It uses per-thread distributed timestamp generation with adaptive backoff, read/write sets for optimistic multi-versioning, and rapid garbage collection.
Note that the original Cicada does not guarantee strict serializability but causal consistency~\cite{DBLP:conf/sosp/PetersenSTTD97}.
We applied epoch-based group commit to this protocol to ensure strict serializability; the original has the 64-bit per-transaction timestamp, but we shortened it to 32-bit and added an epoch number to the upper 32 bits to synchronize between epochs and avoid stale reads.
To test overwriters, we ported Silo's implementation to MVTO+ESSE; when a version $x_i$ in $rs_j$ was not the latest version of data item $x$, MVTO+ESSE failed to test $o_j$ with its erasing version order and delegate its control to MVTO.

\section{Evaluation}
\label{sec:evaluation}

Our experiments used a lightweight, non-distributed, embedded, transactional key-value storage prototype written in C++. It consists of in-memory storage, CC protocols (Silo, Silo+ESSE, MVTO, and MVTO+ESSE), a tree index forked by Masstree~\cite{kohlermasstree}, and a parallel-logging manager according to the SiloR~\cite{Zheng2014FastParallelism}specification.
The experiments were run on a 72-core machine with four Intel Xeon E7-8870 CPUs and 1 TB of DRAM. Each CPU socket had 18 physical cores and 36 logical cores with hyperthreading. The results for over 72 threads showed sublinear scaling due to contention within the physical cores; in the result figures, we use a gray background color to indicate this situation.
Each socket had a 45-MB L3 shared cache.
The transaction logs were separated for each worker thread and flushed into a single solid-state drive. Because all queries were compiled at build time, neither networked clients nor SQL interpreters were used. Each worker thread had a thread-local workload generator that enabled it to input its own transactions.

\begin{figure*}[t]
    \centerline{
        \subfloat
        {
            \includegraphics[width=0.8\textwidth]{benchmark2/artifacts/legend.pdf}
        }
    }
    \vspace{-20pt}
    \addtocounter{subfigure}{-1}
    \centerline{
        \subfloat[Original TATP]
        {\includegraphics[width=0.23\textwidth]{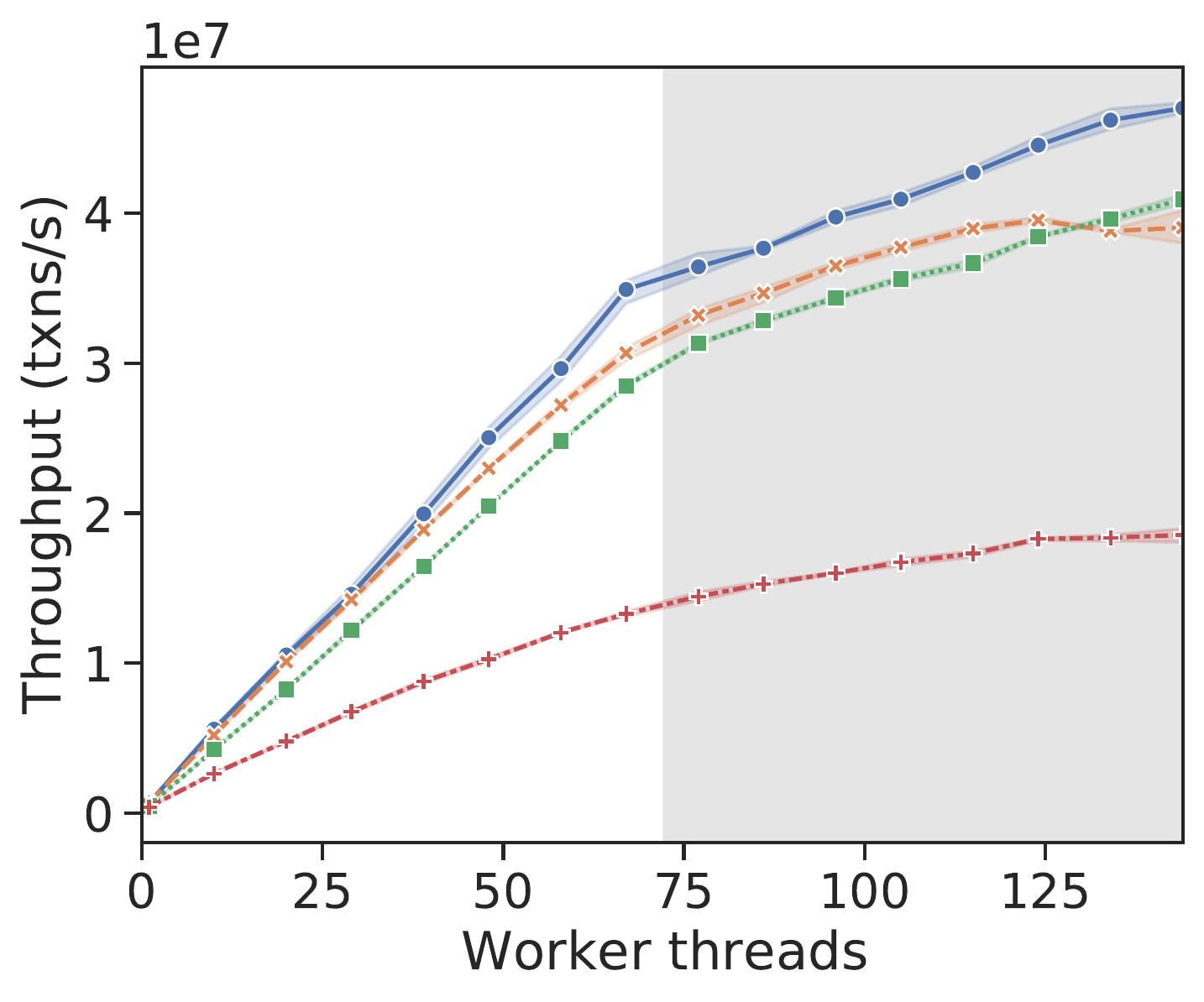}
            \label{fig:tatp-di14}}
        \hspace{2pt}
        \subfloat[Update-intensive modification]
        {\includegraphics[width=0.23\textwidth]{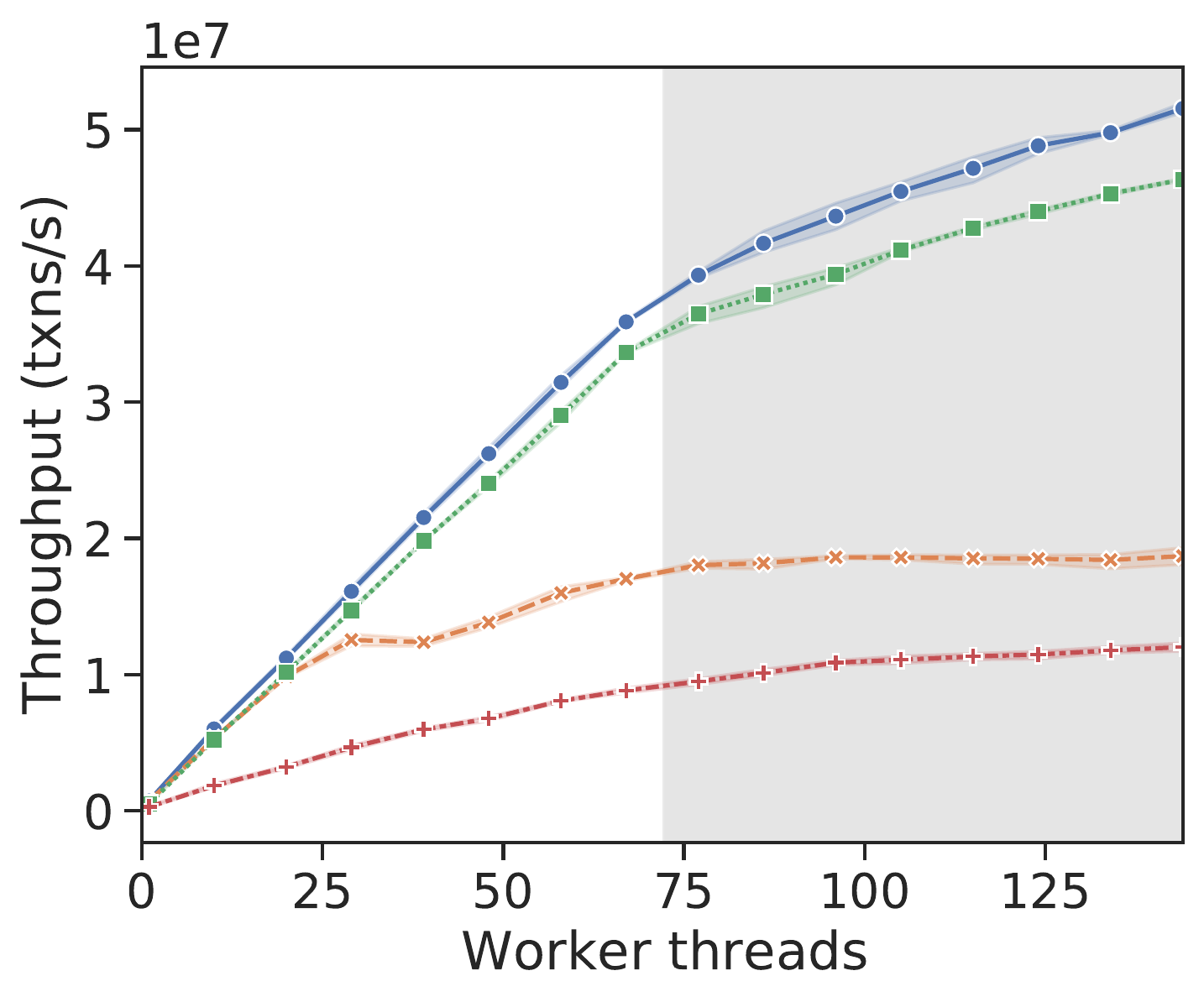}
            \label{fig:tatp-di80}}
        \hspace{2pt}
        \subfloat[Number of updates in (b)]
        {\includegraphics[width=0.23\textwidth]{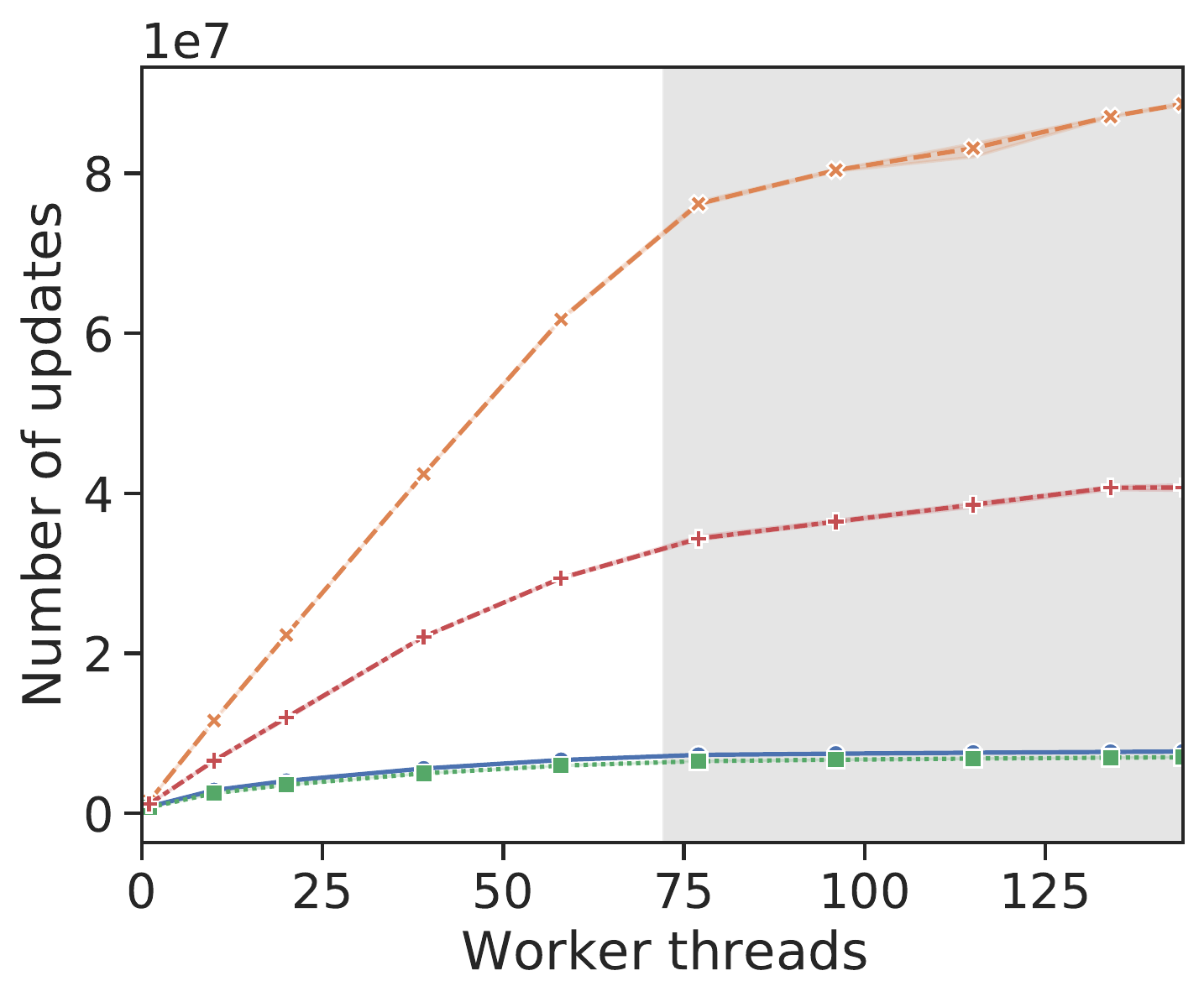}
            \label{fig:tatp-memcomsumption}}
        \hspace{2pt}
        \subfloat[Data ingestion query ratio]
        {\includegraphics[width=0.24\textwidth]{benchmark2/artifacts/TATP-DI-Query-Ratio/result.pdf}
            \label{fig:tatp-varyingDI}}}
    \vspace{0pt}
    \caption{TATP benchmark results}
    \label{fig:tatp}
\end{figure*}

We selected three workloads generated by benchmark specifications: TATP~\cite{TATP}, YCSB~\cite{Cooper2010BenchmarkingYCSB}, and TPC-C~\cite{10.5555/1946050.1946051} benchmarks. We selected TATP as the benchmark for our intended application such as IoT/Telecom applications, including data ingestion queries. We selected YCSB as the ideal scenario for ESSE in terms of performance because it includes a tremendous number of blind updates. Finally, we selected TPC-C as a counterpoint benchmark with the least performance benefit because it includes no blind updates.

\textbf{TATP benchmark.}
TATP represents the workload for a telecommunication company.
It includes 16\% data ingestion queries that generate a flood of blind updates for managing changes in a subscriber's current location or the profile data.
70\% of the rest of queries consist of \texttt{GET\_SUBSCRIBER\_DATA} and \texttt{GET\_ACCESS\_DATA}.
They both retrieve the latest and correct data snapshot updated by data ingestion queries to operate the telecom base station.
We implemented the benchmark in accordance with its specifications.
In addition to the workload obeying the original specifications, we added workloads with various percentages of blind updates from the original's 16\% to emulate our intended IoT/Telecom applications.
The query for which we varied the percentage was \texttt{UPDATE\_LOCATION}.

\textbf{YCSB benchmark.} This workload generator is representative of conventional large-scale online benchmarks. Because the original YCSB does not support a transaction with multiple operations, we implemented a YCSB-like workload generator in our prototype, similarly to DBx1000~\cite{Yu2014StaringCores}.
Specifically, each transaction accessed four data items chosen randomly according to a Zipfian distribution with parameter $\theta$. Each data item had a single primary key and an 8-byte additional column. We populated our prototype implementation as a single table with 100K data items.

\textbf{TPC-C benchmark.} This is an industry-standard for online transaction processing. It consists of six tables and five transactions that simulate the information system of a wholesaler. Note that the TPC-C benchmark does not have any blind updates; all write operations are inserts or read-modify writes. Thus, SSE could not commit any transactions with its erasing version order. We implemented TPC-C full mix including all five transactions.
Phantom anomalies were prevented by the same method with Silo: we scan the tree index again at the commit of each transaction.

\begin{figure}[t]
    \vspace{-12pt}
    \subfloat[With 1 thread]
    {\includegraphics[width=0.45\textwidth]{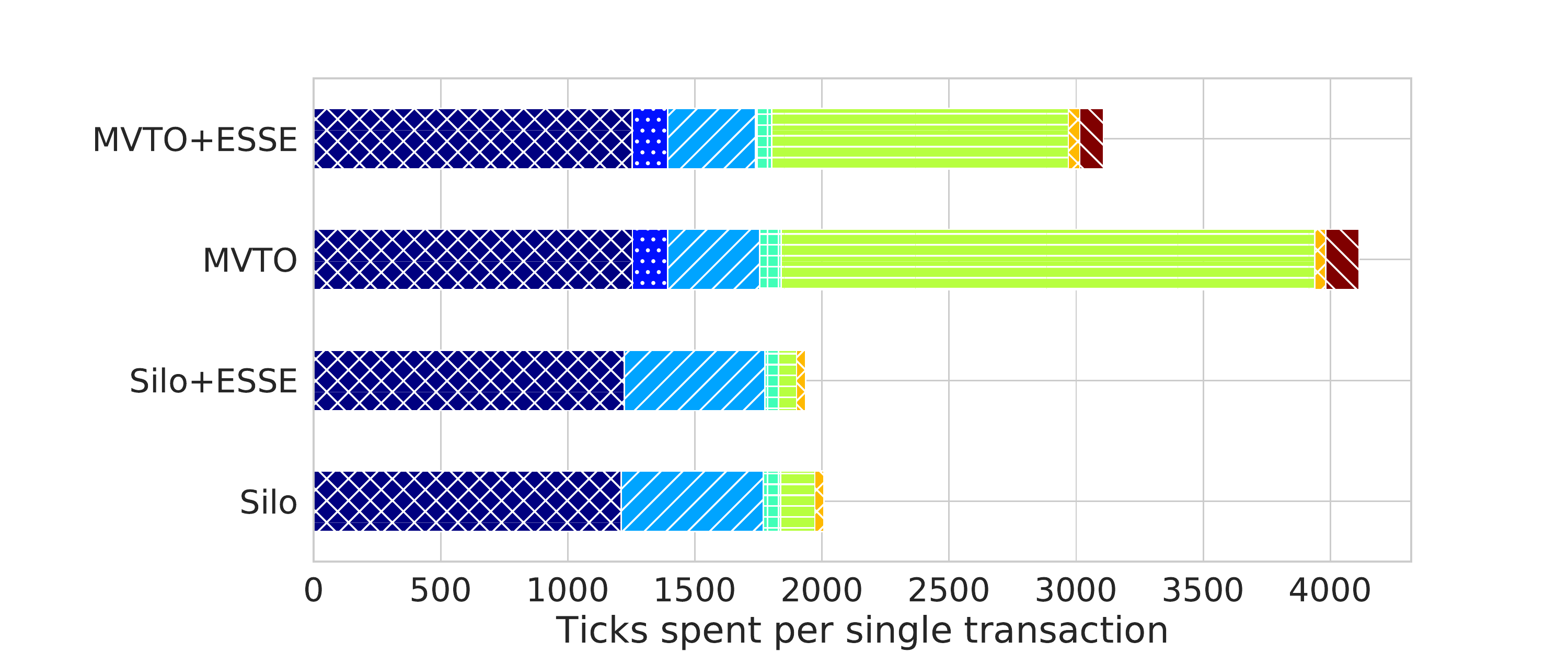}
        \label{fig:tatp-breakdown-di14-1}}
    \vfill
    \vspace{-6pt}
    \subfloat[With 144 threads]
    {\includegraphics[width=0.45\textwidth]{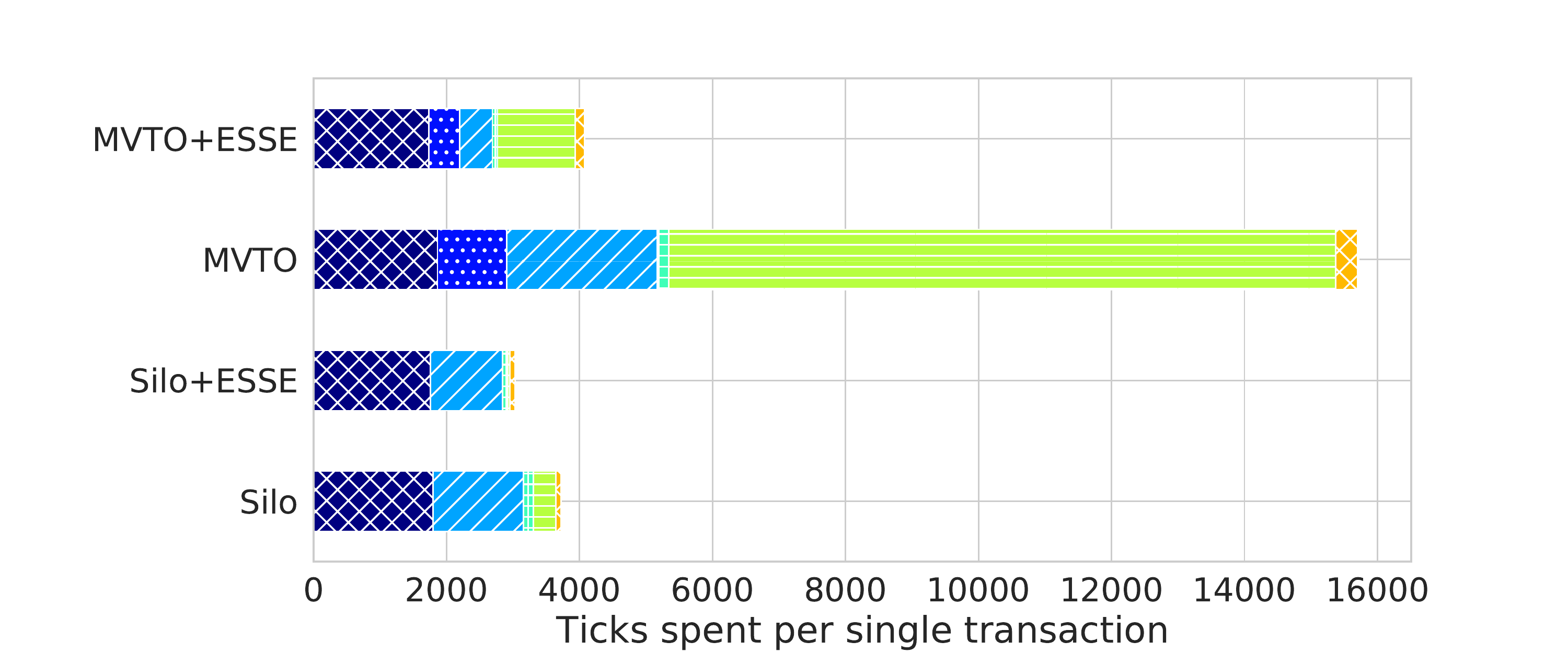}
        \label{fig:tatp-breakdown-di14-144}}
    \vspace{-10pt}
    \subfloat{
        \includegraphics[width=0.55\textwidth]{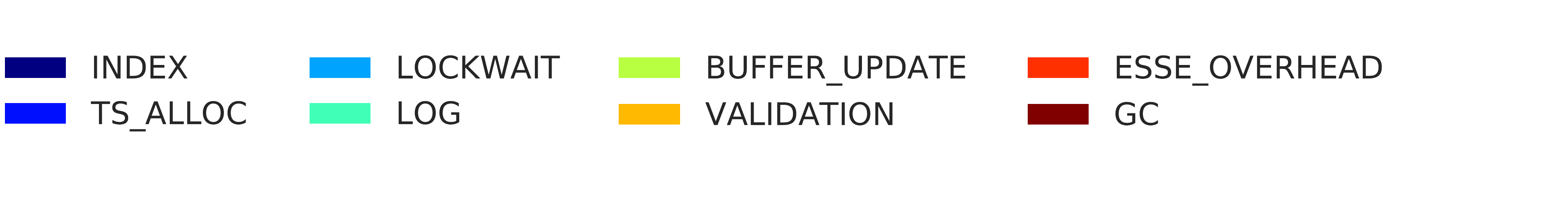}
    }
    \vspace{-10pt}
    \caption{Runtime breakdowns of Figure~\ref{fig:tatp-di14}}
    \label{fig:tatp-breakdown-di14}
    \vspace{-4pt}
\end{figure}

\begin{figure}[t]
    \vspace{-10pt}
    \centering
    \subfloat{
        \includegraphics[width=0.25\textwidth]{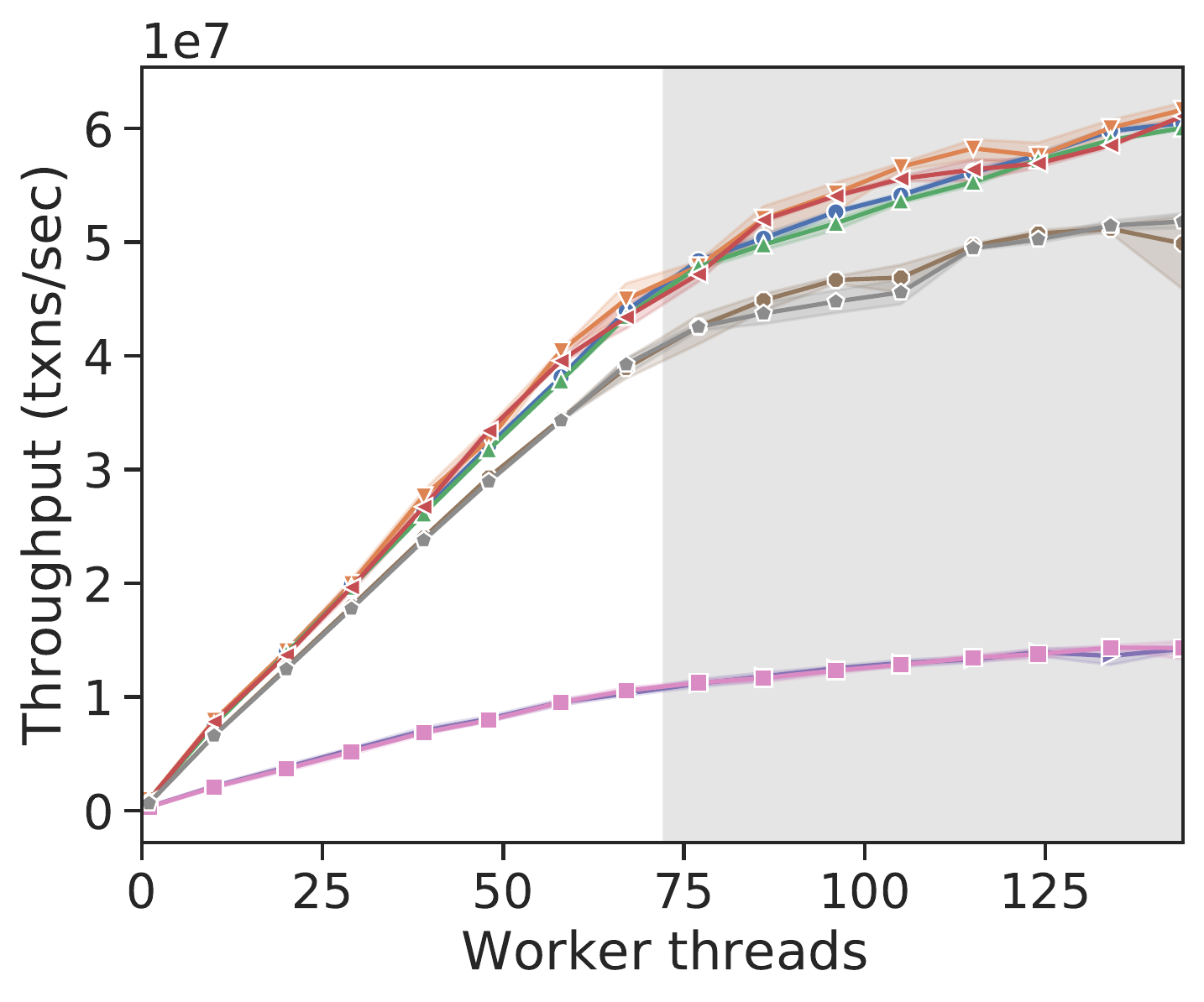}
    }
    \subfloat{
        \raisebox{-15pt}{
            \includegraphics[width=0.18\textwidth]{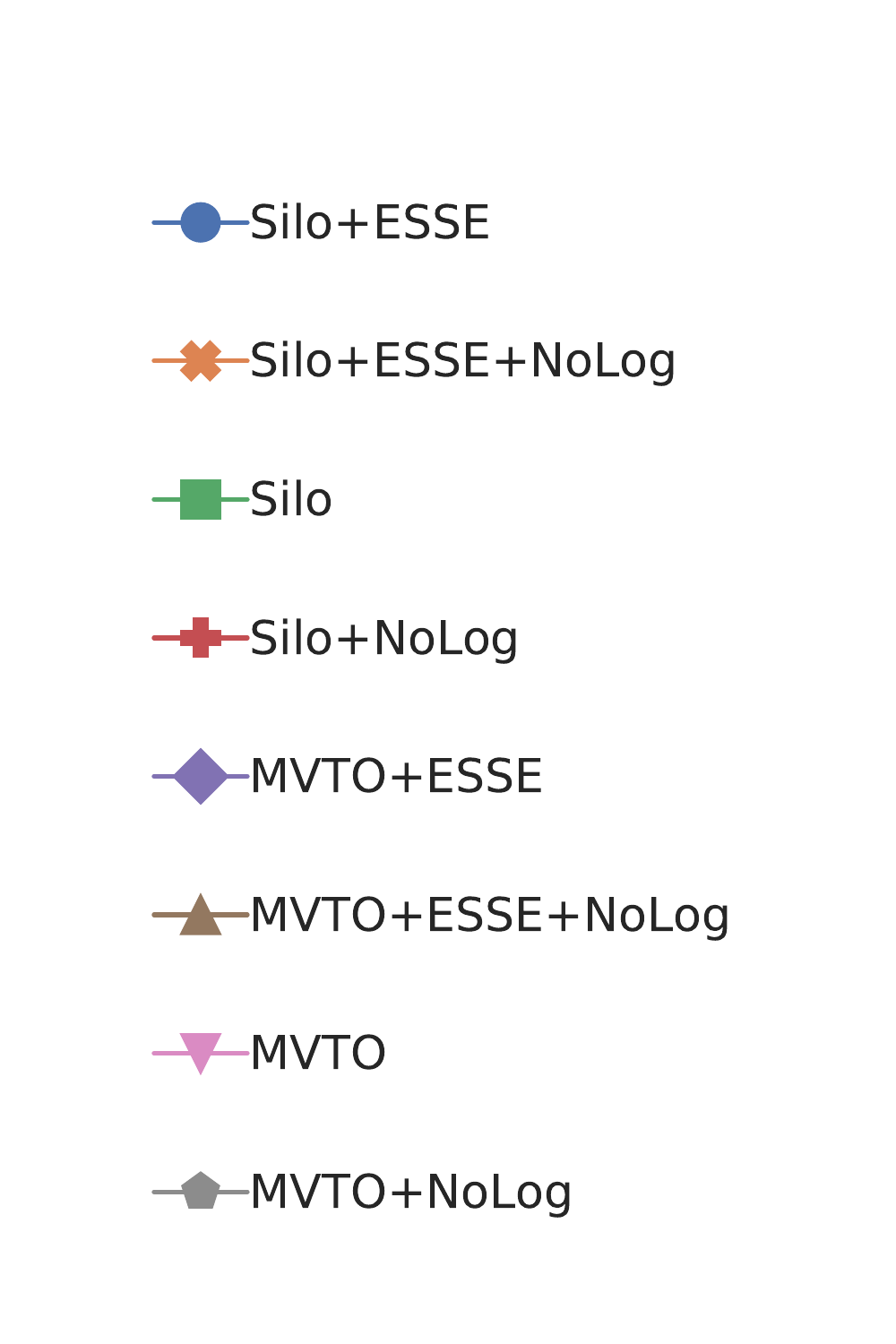}
        }
    }
    \vspace{-12pt}
    \caption{TATP results with logging-disabled protocols}
    \label{fig:tatp-nologs}
    \vspace{-12pt}
\end{figure}


\begin{figure*}[t]
    \centerline{
        \subfloat
        {\includegraphics[width=0.8\textwidth]{benchmark2/artifacts/legend.pdf}}}
    \vspace{-20pt}
    \addtocounter{subfigure}{-1}
    \centerline{
        \subfloat[YCSB-A: scalability]
        {\includegraphics[width=0.24\textwidth]{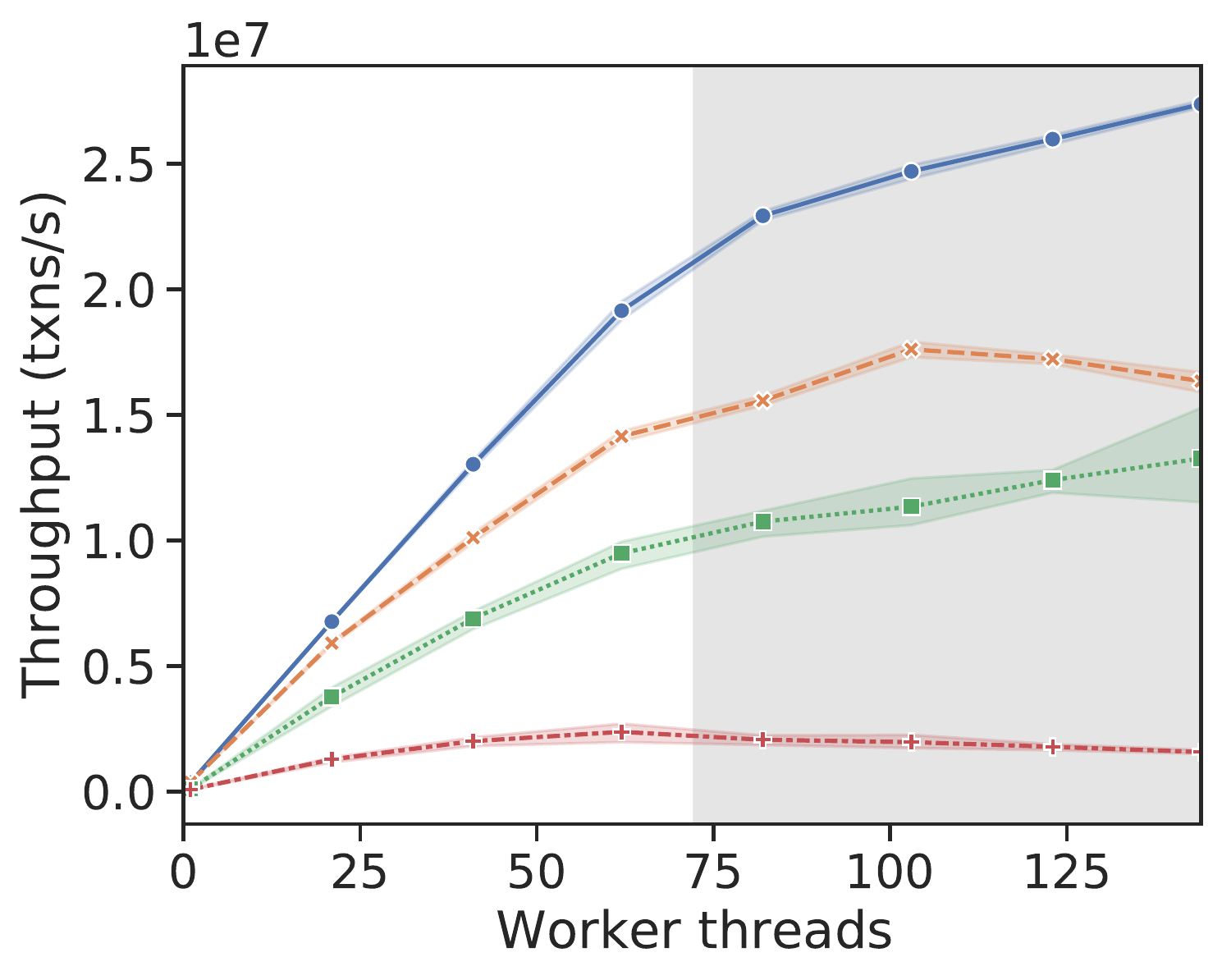}
            \label{fig:ycsb-a-high}}
        \subfloat[YCSB-A: epochs]
        {\includegraphics[width=0.24\textwidth]{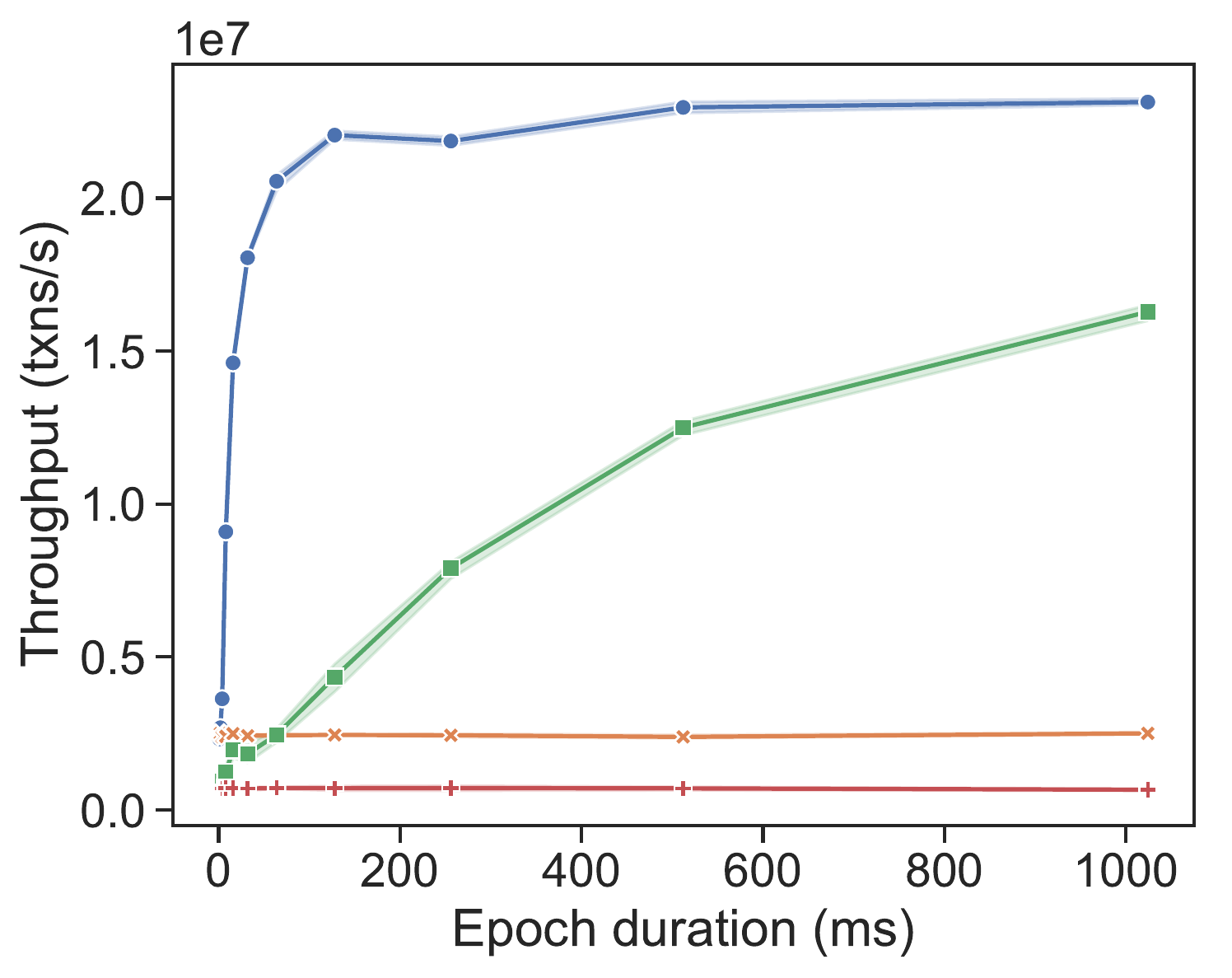}
            \label{fig:ycsb-variable-epoch}
        }
        \subfloat[YCSB-A: size of read/write set]
        {\includegraphics[width=0.23\textwidth]{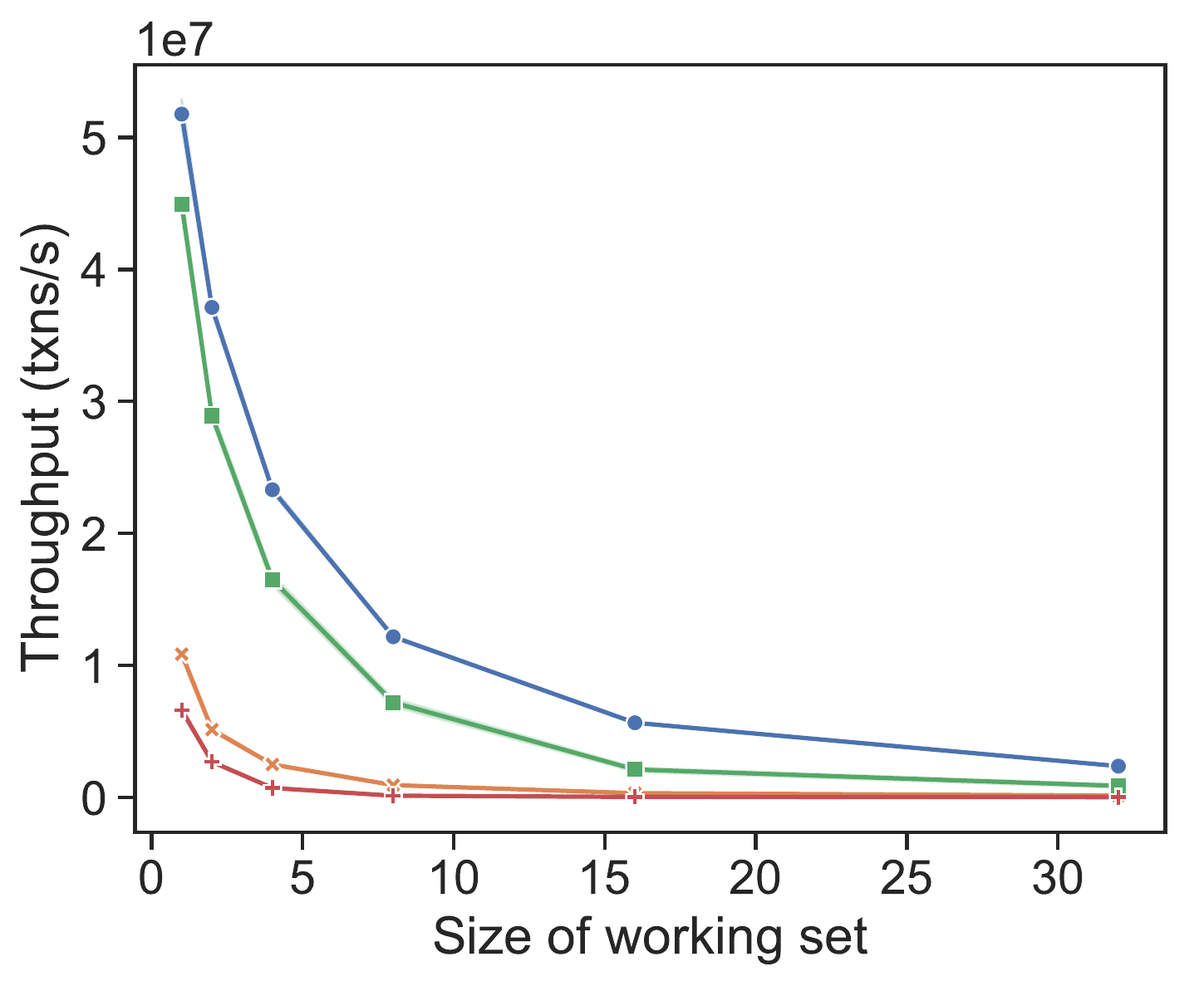}
            \label{fig:ycsb-variable-wssize}}
        \subfloat[YCSB-B: scalability on low contention rate ($\theta=0.2$)]
        {\includegraphics[width=0.24\textwidth]{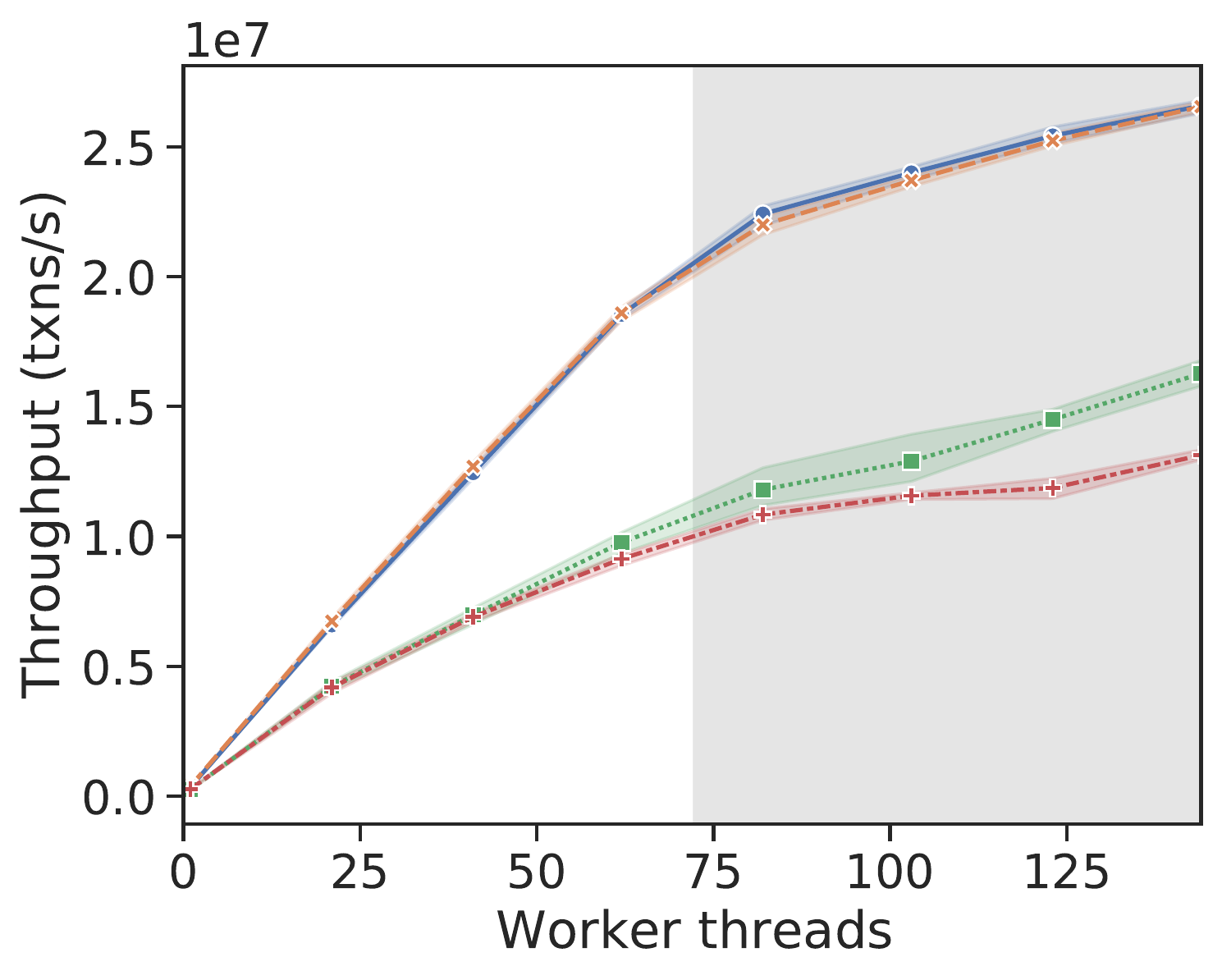}
            \label{fig:ycsb-b-low}}}
    \caption{YCSB benchmark results}
\end{figure*}

\begin{figure}[t]
    \vspace{-20pt}
    \subfloat{
        \includegraphics[width=0.27\textwidth]{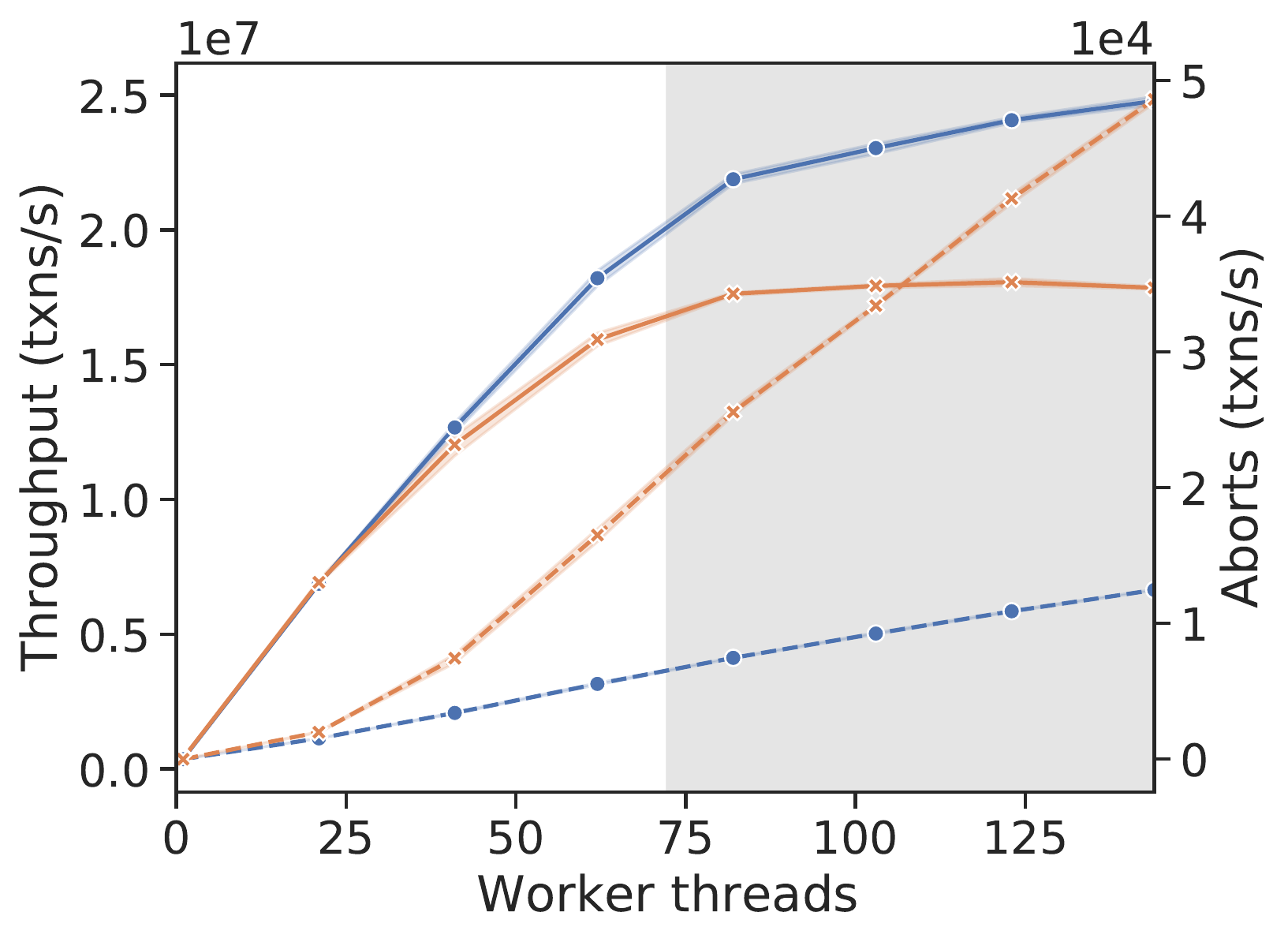}
    }
    \subfloat{
        \raisebox{-40pt}{
            \includegraphics[width=0.20\textwidth]{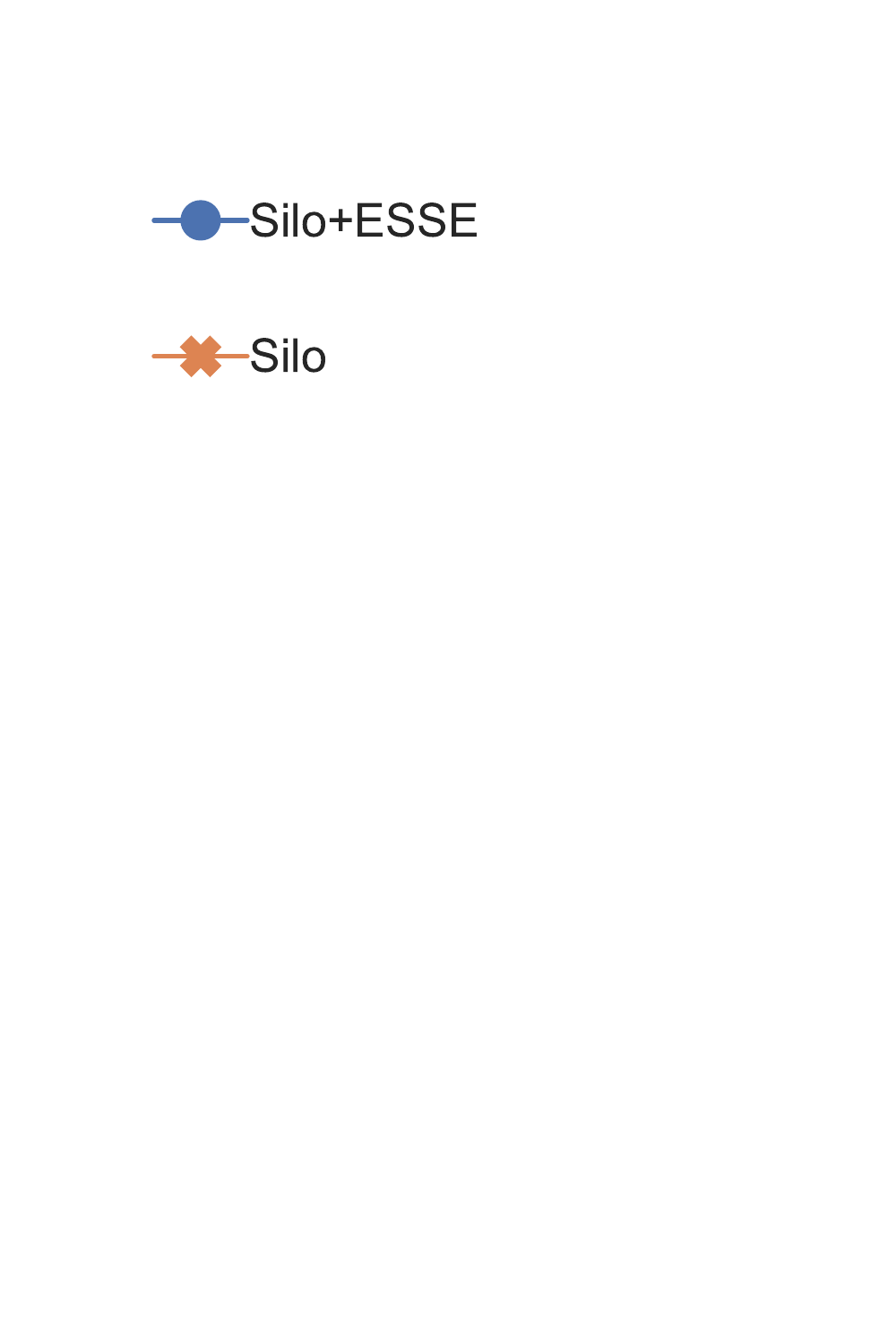}
        }
    }
    \vspace{-40pt}
    \caption{YCSB-B: scalability under high contention ($\theta=0.9$) with throughputs (solid lines) and number of aborts (dashed lines)}
    \label{fig:ycsb-b-high}
    \vspace{-8pt}
\end{figure}


\subsection{TATP Benchmark Results}
\label{sec:tatp}

Figures \ref{fig:tatp-di14} and \ref{fig:tatp-di80} show the results for the original TATP and its update-intensive modification, respectively. In both cases, ESSE improved the performance of the original protocols, and the improvement for the update-intensive modification was particularly remarkable. Because most of the blind updates are safely omittable versions and thus were not installed in physical memory, Silo+ESSE and MVTO+ESSE achieved 2.7$\times$ and 2.5$\times$ performance improvements, respectively. In contrast, the performance of Silo for the update-intensive modification was drastically degraded by lock contention. Although Silo is a read-lock-free protocol, Silo's write operations require lockings that reduce parallelism. MVTO exhibited the poorest performance on both the original TATP and the modification. Although MVTO does not acquire locks for write operations, its throughput degrades because it allocates memory to create new versions. SSE overcame these weaknesses of Silo and MVTO and thus improved the performance drastically.

Figure~\ref{fig:tatp-memcomsumption} shows the number of updates for TATP with the update-intensive modification. For Silo and Silo+ESSE, we counted the number of in-place updates with locks. For MVTO and MVTO+ESSE, we counted the number of out-of-place creating new versions. The results indicate that the two ESSE protocols rarely performed actual updates, because most of the writes generated safely omittable versions.

Figure~\ref{fig:tatp-varyingDI} shows the results for the TATP benchmark with various data ingestion query rates and 144 fixed worker threads. As the percentage of data ingestion queries became larger than the original 16\%, the throughputs of Silo and MVTO dropped. The reasons were that Silo suffered from lock contention on the same data item and MVTO suffered from the management of multiple versions in physical memory. In contrast, the ESSE protocols outperformed the originals and their throughput was not degraded as the percentage increased. Furthermore, when the percentage exceeded 80\%, the performance of these extended protocols improved. This is because SSE provides cache efficiency; in this setting, the clients requested to execute blind writes into almost the same data items, and ESSE thus generated a tremendous number of safely omittable versions. As a result, it rarely installed new versions, and almost all read operations received the same versions that are rarely evicted from the CPU caches.

Figure~\ref{fig:tatp-breakdown-di14} shows the runtime breakdown for the results shown in Figure~\ref{fig:tatp-di14}. With a single thread, the top consumers of CPU ticks were the \texttt{INDEX} block for Silo and Silo+ESSE and the \texttt{BUFFER\_UPDATE} block for MVTO and MVTO+SSE. The overhead of ESSE was negligible for both ESSE protocols. With 144 threads, the primary consumers of CPU ticks were still the same as with a single thread, and the overhead of ESSE was again negligible. Note that ESSE dramatically reduced the number of CPU ticks spent waiting for the \texttt{LOCKING} and \texttt{BUFFER\_UPDATE} blocks. These two overheads were reduced by using the reachability trackers. When a transaction is committed using ESSE's erasing version order, the write operations avoid locks and access only the reachability trackers in a latch-free manner.

As described in Section~\ref{sec:introduction}, our intended applications consist of data ingestion queries and real-time operations. However, there is sometimes a need to support historical analysis~\cite{8509391} for all submitted versions. To support such historical analysis queries, we need to flush the persistent logs for each version as accumulated data. Even if we omit some versions in CC protocols, we also require log persistence for safely omittable versions. Figure~\ref{fig:tatp-nologs} shows the performance results on the TATP benchmark for the protocols with the logging feature disabled. We set the percentage of data ingestion queries to 80\%. We can see that the performance of the ``NoLog'' protocols was almost the same as that of the protocols with logging. This indicates that the performance improvement of the SSE does not come from avoiding log persistence; rather, SSE improves the performance by reducing memory consumption and avoiding lock mechanisms.

\subsection{YCSB Benchmark Results}
\label{sec:ycsb}

Figure~\ref{fig:ycsb-a-high} shows the results for the YCSB-A workload with a medium contention rate ($\theta = 0.6$). YCSB-A defines the proportion of operations as 50\% reads and 50\% blind writes. Prior works showed that conventional protocols perform poorly on such write-contended workloads~\cite{Wu2017AnControl, Yu2014StaringCores, Kim2016Ermia:Workloads, Fan:2019:OVG:3342263.3360357}. Note that YCSB requests the database population before benchmarking; all data items are inserted before measurements. Therefore, in YCSB, all write operations are blind writes. Thus, once the pivot versions for each 40-ms epoch are marked, ESSE's correctness testing rarely fails. We expect that ESSE protocols can avoid installing a tremendous number of blind updates and to improve the performance accordingly. In fact, as shown in Figure~\ref{fig:ycsb-a-high}, the ESSE protocols achieved higher throughput than the original ones. In the best case with 144 threads, the throughput of MVTO+ESSE was more than 20$\times$ better than that of the original protocol. As this workload produces more WAW conflicts than TATP, the length of the linked lists of MVTO tends to be longer. Because the longer linked lists increased the overhead of version traversing for both reads and writes, MVTO's performance was degraded. In contrast, MVTO+ESSE reduced the length of the linked lists because it avoided unnecessary version allocation by omitting blind updates with ESSE.

Figures~\ref{fig:ycsb-variable-epoch} and \ref{fig:ycsb-variable-wssize} show the results for YCSB-A when we varied the two influential parameters for ESSE.
We tested these workloads with 72 threads and $\theta$ = 0.6.
Figure~\ref{fig:ycsb-variable-epoch} shows the results for various epoch sizes.
As described in Section~\ref{sec:implementation}, ESSE uses epoch-based group commits and tests the concurrency among transactions by using epochs. Hence, a longer epoch duration makes more transactions concurrent and reduces the number of ESSE correctness testing failures. Therefore, to investigate the effect of the epoch duration, we tested the YCSB-A with various durations. As shown in Figure~\ref{fig:ycsb-variable-epoch}, the throughput of the ESSE protocols increased with the duration. As a result, we can improve the performance by increasing the epoch duration as much as the application allows.
Next, Figure~\ref{fig:ycsb-variable-wssize} shows the results for various sizes of the read/write sets. 
As the read/write set size increased, the performance of the ESSE protocols decreased to the level of the original protocols. Because the reachability tracker described in Section~\ref{sec:implementation} has two bloom filters (mRS/mWS) with a size of only 16 bits, as the number of operations in a transaction increases, the more often the correctness testing of successors fails because of false positives in the filters.

Figure~\ref{fig:ycsb-b-low} shows the resuls for the read-mostly YCSB-B workload under low contention rate ($\theta = 0.2$).
There are almost no safely omittable versions in this workload since most of the operations are read operations, and contention rarely occurs.
Therefore, there is little benefit from the ESSE's performance improvement; on the contrary, the overhead of ESSE may be painful factor for the performance.
However, Figure~\ref{fig:ycsb-b-low} shows that ESSE protocols performs similar performance to the original protocols.
It indicates the low overhead property of ESSE.
In such workloads, a transaction gives up the ESSE protocol quickly generating an erasing version order, and thus it does not execute the correctness testing.
As in mentioned in Section~\ref{sec:control_flow}, a transaction first checks the epoch number in the indirection of each data item.
ESSE requires all epoch numbers and the transaction’s epoch are the same; however, it is rarely satisfied in this workload.
Therefore, the transaction gives up ESSE's protocol quickly and delegated the control to the baseline protocol.

Figure~\ref{fig:ycsb-b-high} shows the performance results of YCSB-B under high contention rate ($\theta = 0.9$). The dashed lines represent the throughput and the solid lines represent the number of aborts. The YCSB-B workload was expected to be unsuitable for testing our approach because it specifies the proportion of read operations as 95\%~\cite{Cooper2010BenchmarkingYCSB}; SSE and ESSE can improve the performance on blind updates, but YCSB-B rarely executes them. Nevertheless, the ESSE protocols exhibited performance comparable to that of the original protocols, and surprisingly, Silo+ESSE outperformed the original Silo. This improvement indicates that the version omission technique is beneficial for other transactions. In this case, Silo+ESSE changed overwriting operations for the latest versions to omission for stale versions, thus reducing the abort rate because its validation fails when the latest versions are overwritten. In Figure~\ref{fig:ycsb-b-high}, the number of aborts for Silo was higher than the throughput with 25 threads. This was because only 5\% of write operations forced Silo to abort the 95\% of read operations. In contrast, Silo+ESSE kept a lower abort rate than the original Silo, and it improved its performance on this read-mostly workload, because ESSE prevented writes of the latest versions.

\subsection{TPC-C Benchmark Results}
\label{sec:tpc-c}
None of the five queries in the TPC-C benchmark contained blind writes except inserts.
This means that there were no transaction commits with SSE's erasing version order for this workload.
Although our target is IoT/Telecom applications and their data ingestion queries containing a tremendous number of blind updates, we also tested our approach on this benchmark in order to illustrate ESSE's low-overhead property.

To analyze this low-overhead property, we ran the TPC-C benchmark with a single warehouse. This high-contention scenario represents the worst case for ESSE because the reachability tracker for each data item must be frequently updated even though it is never used. Figure~\ref{fig:tpcc-tps} shows the throughput with respect to the number of threads.
Both Silo and Silo+ESSE scale up to 32 cores, similar to the experimental results in the original paper.
MVTO+ESSE's overhead was negligible, yet its performance was comparable to that of the original protocol.
MVTO's performance bottleneck on TPC-C was version traversing or buffer update, so the overhead of the reachability tracker did not affect performance.
In contrast, the throughput of Silo+ESSE was about 0.75$\times$ lower than that of the original protocol~\footnotemark[3].
This was because most of the queries in TPC-C lead to read-modify-writes into the latest versions. Because transactions that perform a read-modify-write operation into a version larger than the pivot version may satisfy the two conditions of Theorem \ref{theo:epoch-barrier}, their footprints must be stored in the reachability tracker for each data item.

\afterpage{
\begin{figure}[t]
    \vspace{-20pt}
    \subfloat{
        \includegraphics[width=0.24\textwidth]{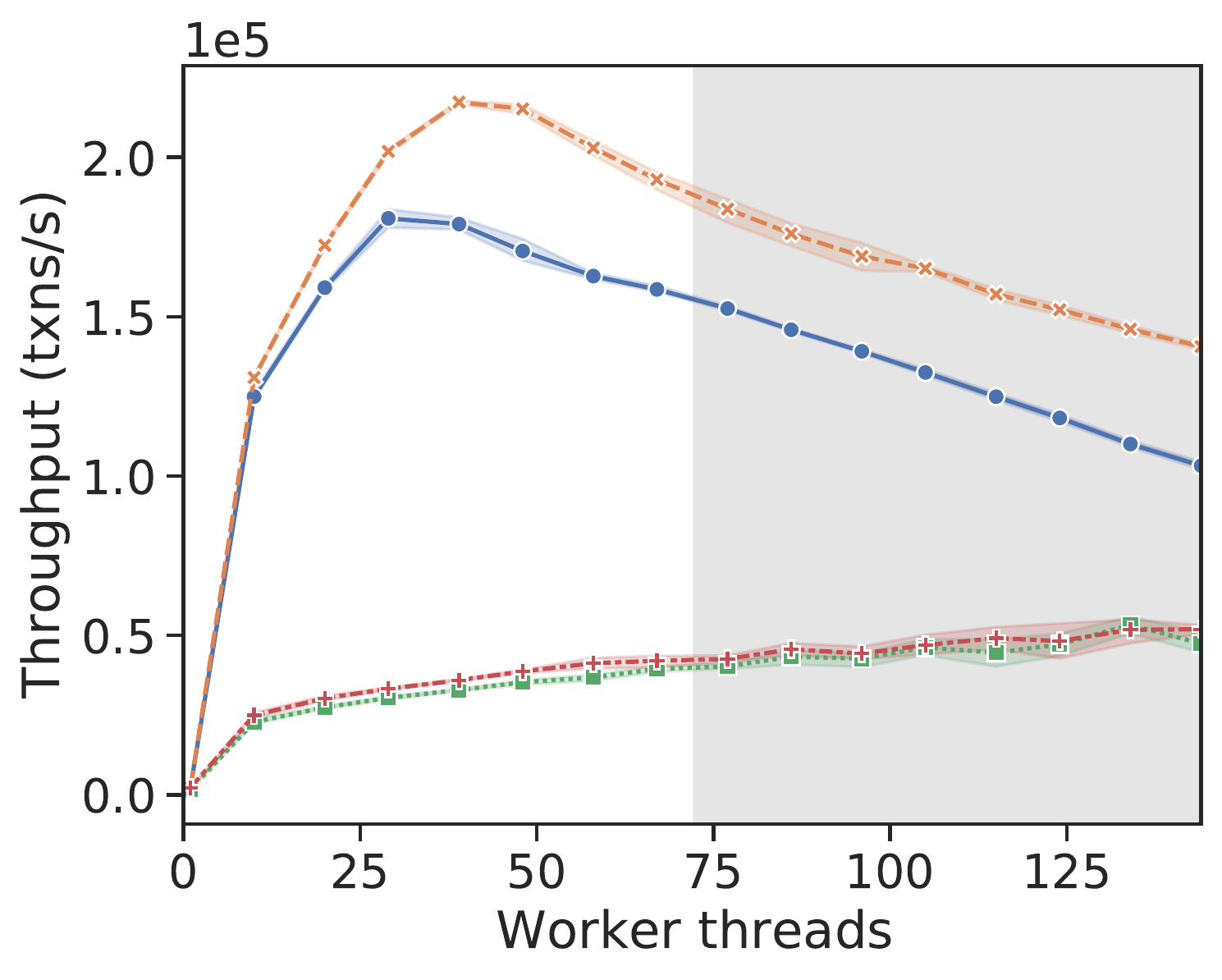}
    }
    \subfloat{
        \raisebox{-40pt}{
            \includegraphics[width=0.20\textwidth]{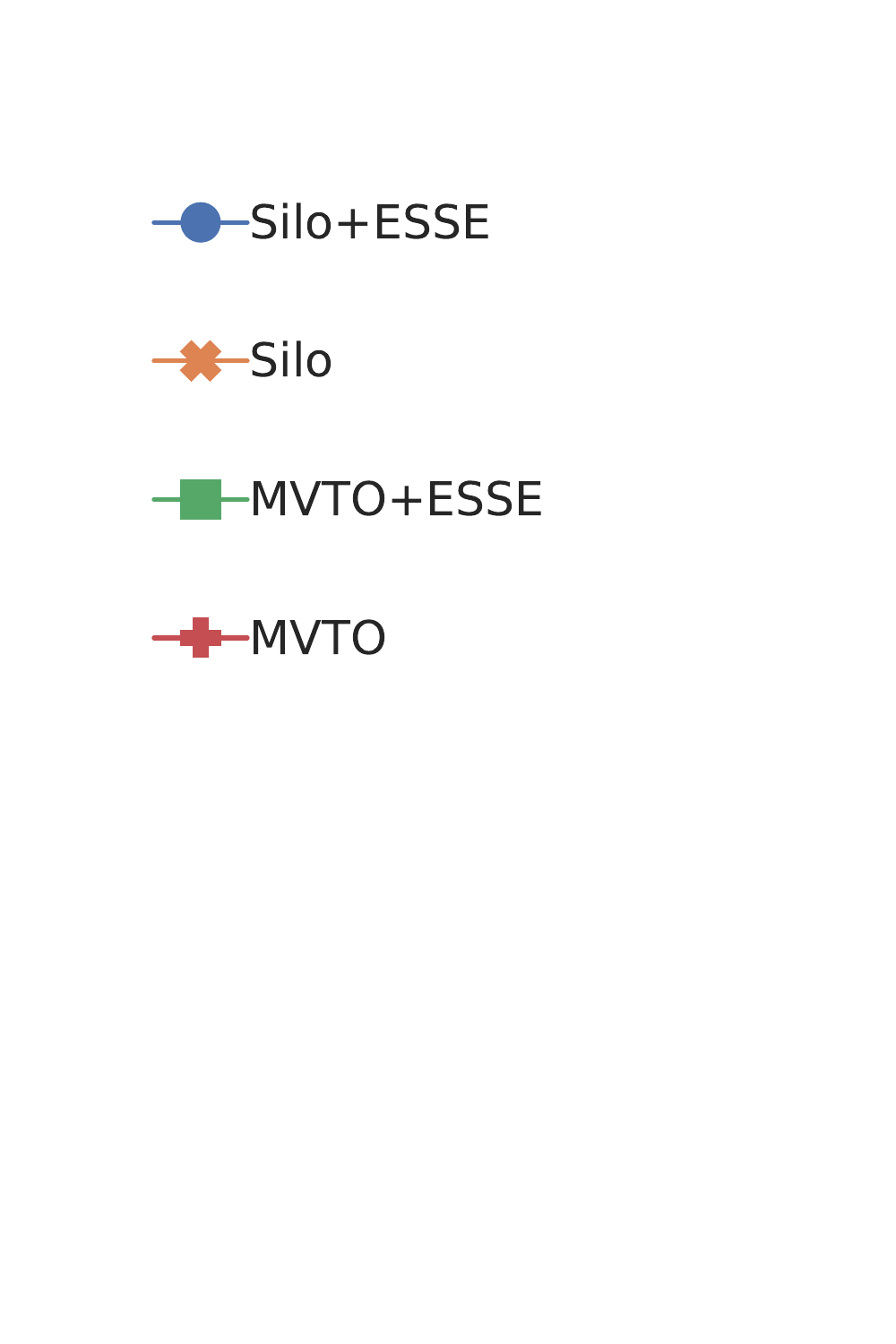}
        }
    }
    \vspace{-40pt}
    \caption{TPC-C: single warehouse}
    \label{fig:tpcc-tps}
    \vspace{-8pt}
\end{figure}

}

\section{Related Work}
\label{sec:relatedwork}

The data ingestion queries are used to aggregate updates from sensors and mobile devices in IoT/Telecom applications.
This query has long been discussed in non-transactional systems such as streaming databases~\cite{DBLP:conf/edbt/GroverC15,DBLP:conf/vldb/TuLPY06}.
In recent years, however, the importance of transactional processing (strict serializability) has been studied~\cite{meehan2017data,DBLP:journals/pvldb/WangC19,10.1145/2882903.2899406}, since reading inconsistent or outdated data causes serious problems in real-world actuators.
We need to process ingested data from sensors and mobile devices efficiently and consistently.
In order to process such a massive number of updates efficiently, it is essential to omit updates that are unnecessary for read operations.
Streaming systems use load shedding~\cite{DBLP:conf/vldb/TuLPY06} or backpressure~\cite{DBLP:journals/tmc/TahirYM18} to put a rate limit of submitting data.
However, these techniques do not ensure transactional correctness; they omit the data before submitting it to databases.
Then we cannot use CC protocols to choose the version order to provide correct data snapshot.
In transaction processing, Thomas's write rule~\cite{Thomas1977ABases} can omit the data and ensure serializability.
However, to the best of our knowledge, TWR does not guarantee strict serializability.
In addition, TWR is applicable only for single-version timestamp ordering protocol, which is obsolete on modern in-memory databases.
SSE also performs the write omission and these methods while preserving the transactional correctness and it is applicable to various CC protocol.

The commutative theory is another example to increase the parallelism on write operations.
Commutative systems such as CRDT~\cite{DBLP:conf/sss/ShapiroPBZ11} and Doppel~\cite{DBLP:conf/osdi/NarulaCKM14}, define a commutative operation set such as ADD or INCR.
Under commutative systems, these operations can be executed in parallel with preserving consistency.
In contrast, SSE focuses on non-commutative operations on the basis of the traditional page-model interface with only READ and WRITE.
As a result, SSE can optimize the performance of IoT/Telecom applications that do not include commutative operations.

Another example of a protocol that generates omittable versions is deterministic databases~\cite{Ren2014AnSystems,Thomson2012Calvin:Systems}. A deterministic database uses the batching approach for concurrency control: centralized transaction managers collect transactions and separate them into batches. Faleiro et al. devised \textit{lazy transaction execution}~\cite{Faleiro2014LazySystems}, in which a deterministic database can generate omittable versions. By delaying the execution of operations, they executes only blind writes that eventually become the latest versions in each batch. After that, database removes other versions. In contrast with lazy evaluation, however, SSE is applicable to non-deterministic protocols because it does not require a centralized transaction manager or prior knowledge of the transactions.

Multi-version concurrency control protocols~\cite{Lim2017Cicada:Transactions, Fekete2005MakingSerializable, Larson2011High-performanceDatabases} can hold multiple versions for each data item.
Multiversion read avoids the high abort rates of single-version protocols, especially for workloads that include long transactions~\cite{Yu2014StaringCores,Wu2017AnControl,Kim2016Ermia:Workloads}.
All MVCC protocols can theoretically generate multiple version orders by using MVSG. However, conventional protocol~\cite{Larson2011High-performanceDatabases,Lim2017Cicada:Transactions,Kim2016Ermia:Workloads,Ports2012SerializablePostgreSQL,Kemper2011HyPer:Snapshots} generates only a single version order.
This is because the decision process to find a suitable one from all possible orders is NP-complete~\cite{Bernstein1983MultiversionAlgorithms,Papadimitriou1982OnVersions}, as mentioned in Section~\ref{sec:validation}. SSE also does not give the exact solution, but it reduces the computational cost by adding only a single candidate version order, which is preferable for generating safely omittable versions.

\section{Conclusion}
\label{sec:conclusion}
We presented a protocol extension method, scheduling space expander (SSE), and its optimized implementation named ESSE. SSE and ESSE can extend various protocol so that it can, in theory, test an additional version order for the purpose of generating safely omittable versions while preserving both strict serializability and recoverability.
To evaluate the performance gain with our approach, we extended two existing protocols, Silo and MVTO to include ESSE.
We expect that SSE and ESSE can help accelerate emerging systems with data ingestion queries.

\bibliographystyle{ACM-Reference-Format}
\bibliography{references}

\end{document}